\documentclass[10pt]{article} % For LaTeX2e
\usepackage[preprint]{tmlr}

\usepackage{amsmath,amsfonts,bm}

\def\eqref#1{equation~\ref{#1}}
\def\1{\bm{1}}

\DeclareMathAlphabet{\mathsfit}{\encodingdefault}{\sfdefault}{m}{sl}
\SetMathAlphabet{\mathsfit}{bold}{\encodingdefault}{\sfdefault}{bx}{n}

\newcommand{\E}{\mathbb{E}}

\DeclareMathOperator*{\argmax}{arg\,max}

\usepackage{caption}
\usepackage{latexsym}
\usepackage{amsmath}
\usepackage{amsthm,amssymb}
\usepackage{booktabs}
\usepackage{enumitem}
\usepackage{graphicx}
\usepackage{color}
\usepackage{subfigure}
\usepackage{hyperref}
\usepackage{url}
\usepackage{algpseudocode}
\usepackage[linesnumbered,ruled]{algorithm2e}

\newtheorem{problem}{Problem}
\newtheorem{observation}{{Observation}}
\newtheorem{claim}{{Claim}}

\newcommand{\Mutate}{\textsc{Mutate }}

\newtheorem{theorem}{Theorem}
\newtheorem{lemma}[theorem]{Lemma}

\newtheorem{definition}{Definition}

\theoremstyle{remark}

\usepackage{soul}

\usepackage{adjustbox}
\usepackage{footnote}
\usepackage{nicefrac}
\usepackage[usestackEOL]{stackengine}
\usepackage{float}
\makesavenoteenv{tabular}
\usepackage{pgfplots}
\usepgfplotslibrary{colorbrewer}

\usepackage[colorinlistoftodos]{todonotes}
\usepackage{tikz} 
\usepackage{multirow}
\usetikzlibrary{arrows}
\tikzset{
    vertex/.style={circle,draw,minimum size=1.5em},
    edge/.style={->,> = latex'}
}
\usetikzlibrary{arrows.meta}
\usepackage{enumitem}
\usepackage{lipsum}

\newcommand{\OPT}{{\rm OPT}}

\title{Improved Evolutionary Algorithms for Submodular Maximization with Cost Constraints}

\author{\name Yanhui Zhu \email yanhui@iastate.edu \\
      \addr Department of Computer Science\\
      Iowa State University
      \AND
      \name Samik Basu \email sbasu@iastate.edu \\
      \addr Department of Computer Science\\
      Iowa State University
      \AND
      \name A. Pavan \email pavan@cs.iastate.edu\\
      \addr Department of Computer Science\\
      Iowa State University}

\begin{document}

\maketitle

\begin{abstract}
We present an evolutionary algorithm {\sc evo-SMC} for the problem of Submodular Maximization under Cost constraints (SMC). Our algorithm achieves $\nicefrac{1}{2}$-approximation with a high probability $1-\nicefrac{1}{n}$ within $\mathcal{O}(n^2K_{\beta})$ iterations, where $K_{\beta}$ denotes the maximum size of a feasible solution set with cost constraint $\beta$. To the best of our knowledge, this is the best approximation guarantee offered by evolutionary algorithms for this problem. We further refine {\sc evo-SMC}, and develop  {\sc st-evo-SMC}. This stochastic version yields a significantly faster algorithm while maintaining the approximation ratio of $\nicefrac{1}{2}$, with probability $1-\epsilon$. The required number of iterations reduces to $\mathcal{O}(nK_{\beta}\log{(1/\epsilon)}/p)$, where the user defined parameters $p \in (0,1]$ represents the stochasticity probability, and $\epsilon \in (0,1]$ denotes the error threshold. Finally, the empirical evaluations carried out through extensive experimentation substantiate the efficiency and effectiveness of our proposed algorithms. Our algorithms consistently outperform existing methods, producing higher-quality solutions. 
\end{abstract}

\section{Introduction}
\label{sec:intro}
A function $f$ defined over a ground set $V$ is submodular if, for any subsets $S$ and $T$ of $V$ where $S \subseteq T$, for any $x \in V \setminus T$, $f(S \cup \{x\}) - f(S) \geq f(T \cup \{x\}) - f(T)$. The function $f$ is monotone if $f(T) \geq f(S)$ when $S$ is a subset of $T$.
Monotone submodular function optimization is a fundamental problem in combinatorial optimization, applied extensively across diverse fields such as feature compression, deep learning, sensor placement, and information diffusion, among others~\cite{bateni2019categorical,el2022data,li2023submodularity,kempe2003maximizing,zhu2024score+}. Over the last decade, various versions of submodular optimization problems have garnered substantial attention.

A typical monotone submodular maximization problem involves finding a set $S \subset V$ such that $f(S)$ is maximized, subject to a constraint that the set $S$ belongs to a family of sets ${\mathcal I}$ known as a feasible family. 
An example of ${\mathcal I}$ is the set of all sets of size at most $k$, known as the cardinality constraint. Although many constrained submodular maximization problems are known to be NP-hard, several variants admit efficient approximation algorithms and have been amenable to rigorous theoretical analysis. 
This work specifically focuses on monotone submodular maximization under cost/knapsack constraints (SMC). Here, in addition to the monotone submodular function $f$, there is a modular cost function $c$ over the ground set $V$. The cost function has the property that $c(S) = \sum_{x \in S} c(x)$. 
The objective of SMC is to identify a set $S$ maximizing $f(S)$ while ensuring that $c(S)$ remains within a prescribed budget $\beta$. When the cost function is uniform, SMC reduces to the classical cardinality constraint submodular maximization problem. \citet{khuller1999budgeted} and \citet{krause2005note} proposed greedy algorithms for SMC that achieve an approximation ration of $\frac{1}{2}(1 - 1/e)$, using $O(n^2)$ calls
to the underlying submodular function $f$. The work of ~\cite{sviridenko2004note} provided an algorithm that achieves a tight $(1-\nicefrac{1}{e})$ approximation ratio for the SMC problem but has a very high time complexity of $O(n^5)$.
Subsequent research efforts have introduced variants of the greedy algorithm aiming to improve runtime at a slight expense of the approximation quality~\cite{feldman2022practical,yaroslavtsev2020bring,li2022submodular,EneN19,badanidiyuru2014fast,abs-2008-05391}. 

Greedy algorithms construct solution sets iteratively by ``greedily" adding one element during each iteration. 
This process continues as long as the underlying constraint is satisfied. While these algorithms can be rigorously analyzed, providing approximation guarantees for solution quality, they do have drawbacks.  Greedy algorithms can get stuck in a local optimum. Moreover, greedy algorithms are fixed-time algorithms -- meaning they can only be executed for a fixed number of iterations; even in scenarios where more computational resources (time) can be afforded, they will not produce higher-quality solutions.

Another approach to submodular optimization involves evolutionary algorithms. These algorithms mimic the population evolution procedure involving random \textit{mutations}. These algorithms behave as follows. They maintain a set of feasible candidate solution sets. During each iteration, a candidate set is randomly selected for mutation. If the mutated set yields a higher-quality solution set (satisfying the constraint), then it replaces the lower-quality solution set. This approach is appealing because the random mutations can aid in moving away from the local optima. Moreover, there is no apriori bound on the number of iterations. If resources permit, the algorithm can be run for a much longer time and potentially could produce higher-quality solutions. 

The  work of~\cite{qian2017subset} used a evolutionary algorithm framework Pareto Optimization to achieve a $\frac{1}{2}(1-\nicefrac{1}{e})$ approximate solution to SMC. However, to guarantee this approximation ratio, the algorithm has to run exponentially many iterations. 
A subsequent work~\cite{bian2020efficient} proposed an evolutionary algorithm EAMC with improved runtime. They proved that when the algorithm EAMC  is run for $O(n^2K_\beta)$ iterations then the produced solution has an approximation ratio of $\frac{1}{2}(1-\nicefrac{1}{e})$. 
Furthermore, empirical results of their work showed that EAMC produces solutions significantly better than those produced by greedy-based approaches. 
The approximation ratio of $\frac{1}{2}(1 - \nicefrac{1}{e})$ that is achieved by these evolutionary algorithms is not competitive compared with the best known $(1-\nicefrac{1}{e})$ approximation ratio by the greedy algorithm (though with a much higher time complexity).  
Therefore, a significant objective is to design novel and efficient evolutionary algorithms that offer stronger approximation ratios.

\subsection{Our Contributions}
We design an evolutionary algorithm {\sc evo-SMC} that achieves an approximation ratio of $\nicefrac{1}{2}$. 
This marks a significant enhancement in the approximation guarantees for evolutionary algorithms applied to SMC. 
The $\nicefrac{1}{2}$ approximation ratio is attained when the algorithm is executed for $\mathcal{O}(n^2K_{\beta})$ iterations, which is still cubic. To address this, we 
refine {\sc evo-SMC} and develop {\sc st-evo-SMC}, which has a $\nicefrac{1}{2}$ approximation with probability $1-\epsilon$ and requires only $\mathcal{O}(nK_{\beta}\ln{(1/\epsilon)}/p)$ iterations. 
In this algorithm, $p \in (0,1]$ is the stochasticity probability that controls the candidate set selection. When $p=0$, {\sc st-evo-SMC} reduces to the original {\sc evo-SMC}.  Algorithm {\sc st-evo-SMC} improves the running time of {\sc evo-SMC} by a magnitude of $n$. Table \ref{table:2} compares the approximation ratios of our algorithms with the state-of-the-art algorithms.

To supplement the theoretical results, we conduct experiments across diverse application domains, such as influence maximization, vertex cover, and sensor placement. The empirical results demonstrate that our algorithms produce higher-quality solutions than the state-of-the-art evolutionary algorithms. The experiments also empirically demonstrate that, when allowed to run for a longer time,  our algorithms perform better than the greedy-based algorithm with the same approximation ratio of $\nicefrac{1}{2}$. As an implementation contribution, we show that the bloom filters \cite{bloom1970space} can be used in evolutionary algorithms to avoid duplicate evaluations resulting in low memory and execution times.

\subsection{Additional Related Work}
The classical work of \cite{nemhauser1978analysis} presented a greedy algorithm that achieves a $(1-\nicefrac{1}{e})$ approximation ratio for the cardinality constraint problem and makes $O(kn)$ calls to the monotone submodular function $f$. The works of~\cite{qian2015subset,friedrich2015maximizing} designed evolutionary algorithms for this problem using the Pareto Optimization framework while achieving the same approximation guarantee of $(1 - \nicefrac{1}{e})$. The subsequent work of~\cite{crawford2019faster} designed evolutionary algorithms that significantly reduced the runtime while achieving an approximation ratio of  $(1-\nicefrac{1}{e}-\epsilon)$.
For the dynamic cost constraints, Pareto Optimization \cite{roostapour2022pareto} and its variant \cite{bian2021fast} are proven to admit $\nicefrac{1}{2}(1-\nicefrac{1}{e})$ approximation.
 The works in~\cite{iyer2013submodular,padmanabhan2023maximizing} considered the scenario where the function  $c$ is also submodular. A related problem to SMC is  maximizing $f-c$, which has been studied in \cite{harshaw2019submodular,jin2021unconstrained,qian2021multiobjective}.   Evolutionary algorithms have been proposed for a few other variants of submodular optimization problems \cite{chen2022budgeted,do2021pareto}.

\paragraph{Organization.}
Preliminary definitions and notations are given in Section~\ref{sec:prelim}. The algorithms and analysis for solving Problem 1 are presented in Section \ref{sec:po} and \ref{sec:bpo}. Section \ref{sec:experiements} discusses applications, experimental setup, datasets, baseline algorithms and results. Finally, we conclude with a discussion in Section \ref{sec:conclusion}. Theoretical analysis of the proposed algorithm in Section \ref{sec:bpo} and additional experimental results are deferred to the appendix.

\begin{table}[t]
\caption{Approximation ratios comparison with SOTAs for Problem \ref{pro:smk}. }
\centering
\begin{tabular}{c c c c c }
    \hline
    Algorithm & Approximation Ratio    \\ [0.5ex] 
    \hline
    {\sc POMC} (\cite{qian2017subset}) & $\nicefrac{1}{2}(1-\nicefrac{1}{e})$  \\
    {\sc EAMC} (\cite{bian2020efficient}) & $\nicefrac{1}{2}(1-\nicefrac{1}{e})$  \\
    {\sc evo-SMC} (this work) & $\nicefrac{1}{2}$  \\
    {\sc st-evo-SMC} (this work) & $\nicefrac{1}{2}$ \\
    \hline
\end{tabular}
\label{table:2}
\end{table}

\section{Preliminaries}
\label{sec:prelim}
Given a ground set $V$ of size $n$ and a set function $f:2^V\rightarrow \mathbb{R}^+$, $f$ is monotone if for any subsets $S \subseteq T \subseteq V$, $f(T) \geq f(S)$. The {\em marginal gain} of adding an element $x \in V \setminus S$ into $S$ is $f(\{x\} \mid S) \triangleq f(S\cup \{x\})-f(S)$. We write $f(\{x\} \mid S) = f(x \mid S)$ for for brevity and assume that $f$ is normalized, i.e., $f(\emptyset)=0$. 
A function $f$ is submodular if for any set $S \subseteq T \subseteq V$ and any element $x \in V \setminus T$, the  marginal gain of the function value of adding $x$ to $S$ is at least the marginal gain in the function value of adding $x$ to a larger set $T$. This property is referred to as {\em diminishing return}. Formally, 
\begin{equation}
\label{eq:marginal-gain}
    f(S \cup \{x\})-f(S)\geq f(T \cup \{x\})-f(T).
\end{equation}

We consider a modular cost function $c: 2^V \rightarrow \mathbb{R}^+$ over a set $S \subseteq V$ where $c(S) = \sum_{s \in S}c(s)$. Note that a function is modular iff the equality of Eq.~(\ref{eq:marginal-gain}) holds.  The {\em density} of an element $e$ adding to set $S$ is defined to be $\frac{f(e|S)}{c(e)}$. Define the maximal size of a feasible set as $K_\beta = \max\{|X|: X \subseteq V \wedge c(X)\leq \beta\}$.

\begin{problem}[Submodular Maximization with Cost constraints (SMC)]
    Given a monotone submodular function $f:2^V\rightarrow \mathbb{R}^+$ and a modular cost function $c:2^V\rightarrow \mathbb{R}^+$, a budget $\beta$, find a subset $X \subset V$ such that
    $\argmax_{X \subset V, c(X) \leq \beta} f(X)$.
    \label{pro:smk}
\end{problem}

The Chernoff bound is an exponentially decreasing upper bound on the tail of random variables. We present the form of the definition that will be used in this paper.
\begin{lemma}[Chernoff Bound]
    Consider independent 0-1 variables $Y_1, \cdots, Y_T$ with the same expectations (means) and $Y=\sum_{i=1}^T Y_i$, if $\E[Y]=\mu$, for a real number $\eta \in (0,1)$, we have
    \[\Pr[Y \leq (1-\eta)\mu] \leq e^{-\eta^2\mu/2}.\]
    \label{lem:chernoff}
\end{lemma}

\paragraph{Oracles.} In this paper, our algorithms are designed under the value oracle model, which is based on the assumption that there exists an oracle capable of returning the value $f(S)$ with a set $S$. The specific value of $f(S)$ can be the expected influence spread in influence diffusion~\cite{kempe2003maximizing}, the covered vertices in directed vertex coverage, or entropy in sensor placement. Our computational complexity analyses and experiments are based on oracles.

\section{An Evolutionary Algorithm {\sc evo-SMC} }
\label{sec:po}
In this section, we introduce an evolutionary algorithm {\sc evo-SMC} for Problem \ref{pro:smk} along with its theoretical guarantee and analysis. 
We proceed with defining an auxiliary function $g$, referred to as \emph{surrogate function} in~\cite{bian2020efficient} that is useful in this algorithm. 
\begin{definition}[Surrogate function $g(X)$] For a set $X \subseteq V$,
\[ g(X) =
  \begin{cases}
    f(X)/c(X)       & \quad |X|\geq 1\\
    f(X)  & \quad |X| = 0.
  \end{cases}
\]
\label{def:g}
\end{definition}
The pseudocode of {\sc evo-SMC} is presented in Algorithm \ref{alg:po}. The algorithm maintains three sets of candidate solutions $F,G,G'$. The $i$-th 
element of each $F$ and $G$ is of size $i$, and are denoted by $F_i$ and $G_i$. $G'_i$ is an augmentation of set $G_i$ using a maximal marginal gain element (inspired by \cite{yaroslavtsev2020bring}).
The candidate solutions conform to the cost constraint for the problem. The algorithm takes an input $T$, the loop controlling parameter describing the number of iterations, and outputs the set for which the valuation of the objective function $f$ is maximal.    

At every iteration, the algorithm randomly selects a set $S$ from the candidate solution sets $F$ and $G$ and ``mutate" $S$ (lines 3-4). The \Mutate procedure uniformly at random flips the membership of the elements in $S$ with probability $\nicefrac{1}{n}$. More specifically, for every element in $V$, if it appears in $S$, remove it from $S$ with probability $\nicefrac{1}{n}$, otherwise add it to $S$ with probability $\nicefrac{1}{n}$. If the mutated set $S'$ satisfies the budget constraint, the algorithm compares $S'$ with a set in the solution pools $F$ and $G$ with the {\em same} size as $S'$. The corresponding set(s) in the solution pools will be updated if $S'$ is better w.r.t. $g$ value or $f$ value. The algorithm also considers an augmented solution $G'$ with the largest marginal gain (lines 11-14), which will be necessary for the theoretical analysis. 

\begin{algorithm}[t]
    \SetKwInOut{Input}{Input}
    \SetKwInOut{Output}{Output}
    \Input{$f:2^V \rightarrow \mathbb{R}_{\geq 0}, c: 2^V \rightarrow \mathbb{R}_{\geq 0},$ total number of iterations $T\in \mathbb{Z}_{> 0}, $ cost constraint $\beta \in \mathbb{R}_{> 0}$}
    \Output{$\argmax_{X \in \{F_0,\cdots,F_{n},G_0, \cdots, G_{n},G_0', \cdots, G_{n}'\}}f(X)$}
    $F_j \gets \emptyset, G_j \gets \emptyset, G_j' \gets \emptyset$ for all $j \in [0, n-1]$\\
    \For{$t \gets 1$ \textbf{to} $T$} {
        $S \gets \text{Random}(\{F_0,\cdots,F_{n-1},G_0, \cdots, G_{n-1}\})$\\
        $S' \gets$ \Call{\Mutate}{$S$}\\
        $i \gets |S'|$\\
        \If{$c(S') \leq \beta$}{
            \If{$f(F_i) < f(S')$}{
                $F_i \gets S'$
            }
            \If{$g(G_i) < g(S')$}{
                $Q \gets G_i \cup \{v\}$ s.t.
                $v=\argmax_{e \in V\setminus G_i,c(e) \leq \beta - c(G_i)}f(e \mid G_i)$\\
                \If{$f(Q) > f(G'_i)$}{$G'_i \gets Q$}
                $G_i \gets S'$
            }
        }
    }
    \caption{{\sc evo-SMC}}
    \label{alg:po}
\end{algorithm}

\begin{theorem}
    Given a monotone submodular set function $f$, a modular cost function $c$, a cost constraint $\beta$, let $\OPT_{\beta}=\argmax\{f(X) \mid X\subseteq V, c(X)\leq \beta\}$. If $K_{\beta}=\max\{ |X| \mid X \subseteq V \wedge c(X) \leq \beta \}$, then after $T\geq \max\{4en^2K_{\beta}, 16en^2\log n\}$ iterations,  Algorithm \ref{alg:po} outputs $X$ such that with probability $1-\frac{1}{n}$,
    \[f(X) \geq \nicefrac{1}{2} \cdot f(\OPT_{\beta}).\] 
\label{th:po}
\end{theorem}

In order to prove the above theorem, we will introduce an auxiliary variable $\omega$, which keeps track of a
\emph{good mutation}. We proceed with the definition of $\omega$ and followed by a necessary lemma for the proof of Theorem \ref{th:po}. 

\begin{definition}[Good mutation and $\omega$]
The valuation of $\omega$ at the end of iteration $t$ is denoted by $\omega(t)$. 
We say that $\omega(0) = 0$, i.e., $\omega$ initially is $0$. 
$\omega(t) = \omega(t-1) + 1$ is incremented iff the following two conditions are met during an iteration $t$
\begin{enumerate}
    \item A specific set $S=G_{\omega(t-1)}$ is selected for mutation at line 3 of Algorithm \ref{alg:po}.
    \item There is exactly one element added to $S$ during \Mutate procedure, nothing else, say $S'=S \cup \{a\}$, where $a = \argmax_{e\in \OPT_{\beta} \setminus (S \cup \{o^*\})} \frac{f(e|S)}{c(e)}$, where $o^* = \argmax_{e \in \OPT_{\beta}} c(e)$.
\end{enumerate}
The above conditions collectively correspond to a good mutation. In other words, the valuation of $\omega$ is incremented at the end of iteration $t$ iff a good mutation occurs. 
\label{def:w}
\end{definition}
 
At the end of iteration $t$, we define $X_{t}$ to be the set where $X_{t} = G_{\omega(t)}$.
Note that, if for iterations $i$ and $j$ such that $\omega(i)=\omega(j)$ and $i>j$, then $f(X_i) \geq f(X_j)$ and $g(X_i) \geq g(X_j)$, 
which are guaranteed by lines 7-14 of Algorithm \ref{alg:po}.

Next, we present the lemma that forms the basis for the proof of the Theorem~\ref{th:po}.
\begin{lemma}
    In Algorithm \ref{alg:po}, at the beginning of an iteration $t$, if the set $S$ is selected for mutation such that 
    \[
    c(S) \leq c(\OPT_{\beta})-c(o^*) \mbox{ where } o^* = \argmax_{e \in \OPT_{\beta}} c(e),
    \] and if the algorithm adds the element $v$ to $S$ where $v = \argmax_{e \in \OPT_{\beta} \setminus (S \cup \{o^*\})} \frac{f(e|S)}{c(e)}$, 
    then,
    \[ f(v \mid S) \geq \frac{c(v)}{c(\OPT_{\beta})-c(o^*)}\left(f(\OPT_{\beta})-f(S\cup \{o^*\})\right). \]
\label{lem:greedy}
\end{lemma}

\begin{proof}

Since $v = \argmax_{e \in \OPT_{\beta} \setminus (S \cup \{o^*\})} \frac{f(e|S)}{c(e)}  $, for every $u\in \OPT_{\beta}\setminus(S\cup \{o^*\})$, we have
\begin{equation}
    \label{eq:greedy-marginal}
    \frac{f(v \mid S)}{c(v)} \geq \frac{f(u \mid S)}{c(u)}.
\end{equation}
Next consider the following inequalities.
    \begin{align*}
        &f(\OPT_{\beta})-f(S \cup \{o^*\}) \\
        &\leq f(\OPT_{\beta} \cup S \cup \{o^*\} ) - f(S \cup \{o^*\}) \quad\mbox{\emph{due to monotonicity}}\\
        & \leq \sum_{u \in \OPT_{\beta}\setminus(S \cup \{o^*\})} f(u | S \cup \{o^*\} ) \quad\quad \mbox{\emph{due to submodularity}}\\
        & \leq \sum_{u \in \OPT_{\beta}\setminus(S \cup \{o^*\})}  f(u | S ) \quad\quad\quad \mbox{\emph{due to submodularity}}\\
        & \leq \sum_{u \in \OPT_{\beta}\setminus(S \cup \{o^*\})} \frac{f(v|S)}{c(v)} \cdot c(u) \quad\quad \mbox{\emph{from Eq. (\ref{eq:greedy-marginal})}}\\
        &\leq \frac{f(v|S)}{c(v)} \left(c(\OPT_{\beta})-c(o^*)\right).
    \end{align*}
Rearranging the terms concludes the proof of Lemma \ref{lem:greedy}.
\end{proof}

To finish the approximation guarantee proof of Theorem \ref{th:po}, we next prove the following lemma.

\begin{lemma}
    For any iteration $t \in [1, T]$, let $X_t$ be the set with size $\omega(t)$ at the end of iteration $t$, then we have:
    \[ f(X_t) \geq \frac{c(X_t)}{2(c(\OPT_{\beta})-c(o^*))}f(\OPT_{\beta}). \]
\label{lem:Xt}
\end{lemma}
\begin{proof}
    We prove this lemma by regular induction.

    \medskip
    \noindent
    \textbf{Base Case: } For the first iteration $t=1$ of Algorithm \ref{alg:po}, we have two cases B-1 and B-2.

    \noindent
     \textbf{Case B-1}: $\omega$ is not incremented, then we have $\omega(1)=\omega(0)=0$. By the definition of $X_t$, $|X_0|=|X_1|=\omega(1)=0$ and $X_1=\emptyset$. So $X_1$ holds for Lemma \ref{lem:Xt} by the following inequality.
        \[f(X_1)=0 \geq \frac{c(X_1)}{2(c(\OPT_{\beta})-c(o^*))}f(\OPT_{\beta})=0.\]
    \noindent
     \textbf{Case B-2}: $\omega$ is incremented, then we have $\omega(1)=\omega(0)+1=1$. Since this is the first iteration, $X_0=\emptyset$, $|X_1|=1$. Assume that the element in condition 2 of Definition \ref{def:w} is $v$. Then, $X_1=\{v\}$, and by applying Lemma \ref{lem:greedy}, we have
        \begin{align*}
            &f(X_1) = f(X_1) - f(X_0) = f(\{v\}) - f(\emptyset) \\
            &\geq \frac{c(v)}{c(\OPT_{\beta})-c(o^*)}(f(\OPT_{\beta})-f(\{o^*\})). \quad \mbox{\emph{due to Lem. \ref{lem:greedy}}}
            \end{align*}
        If $f(\{o^*\}) \geq \frac{1}{2}f(\OPT_{\beta})$, then as per the lines~11-14 of Algorithm~\ref{alg:po}, it is guaranteed to 
        admit a $\nicefrac{1}{2}$ approximation. Therefore, we focus on the case where $f(\{o^*\}) < \frac{1}{2}f(\OPT_{\beta})$.
        Using this relation, we have
         \begin{align*}   
         f(X_1) 
            & \geq \frac{c(v)}{(c(\OPT_{\beta})-c(o^*))}\times \frac{f(\OPT_{\beta})}{2}\\
            &= \frac{c(X_1)}{2(c(\OPT_{\beta})-c(o^*))}f(\OPT_{\beta}) \quad\mbox{\emph{due to $X_1=\{v\}$.}}
        \end{align*}
    Thus, Lemma \ref{lem:Xt} holds for the base case ($t=1$). 
    Next, we prove that for iterations $t\in(1, T]$, Lemma \ref{lem:Xt} still holds. 
    
\noindent        
\textbf{Induction Steps: } For iterations $t>1$, the induction hypothesis (I.H.) is: at the end of iteration $t-1$, we have
\begin{equation}
    \label{eq:IH-cpo}
    f(X_{t-1}) \geq \frac{c(X_{t-1})}{2(c(\OPT_{\beta})-c(o^*))}f(\OPT_{\beta}).
\end{equation}

For every iteration $t>1$, we analyze $X_t$. We also have two cases I-1 and I-2 based on whether $\omega$ is incremented. \\
\smallskip
    \noindent
    \textbf{Case I-1:} If $\omega$ is not incremented, then we have $\omega(t)=\omega(t-1)$, and $|X_t|=|X_{t-1}|=\omega({t-1})$.
    Furthermore, the Algorithm~\ref{alg:po} (lines 7-14) ensures that $g(X_t) \geq g(X_{t-1})$. 
    Therefore, by Definition \ref{def:g}, we can derive
    \begin{equation}
        g(X_t) \geq g(X_{t-1}) \Rightarrow \frac{f(X_t)}{c(X_t)} \geq \frac{f(X_{t-1})}{c(X_{t-1})}
        \label{eq:I-1}.
    \end{equation}
    
    Rearranging the above inequality, we have
    \begin{align*}
        f(X_t) &\geq \frac{c(X_t)}{c(X_{t-1})}f(X_{t-1}) \quad\quad\mbox{\emph{due to Eq. \ref{eq:I-1}}}\\
        &\geq \frac{c(X_t)}{c(X_{t-1})}\cdot \frac{c(X_{t-1})}{2(c(\OPT_{\beta})-c(o^*))}f(\OPT_{\beta}) \quad\mbox{\emph{due to Eq.~(\ref{eq:IH-cpo})}}\\
        &=\frac{c(X_t)}{2(c(\OPT_{\beta})-c(o^*))}f(\OPT_{\beta}).
    \end{align*}
    This concludes the proof of Lemma \ref{lem:Xt} for Case I-1.  
    
    \noindent
    \smallskip
    \textbf{Case I-2:} If $\omega$ is incremented, then we have $\omega(t)=\omega(t-1)+1$. In this case, $X_{t-1}$ is selected for mutation and the \emph{good} mutation results in $|X_t|=|X_{t-1} \cup \{v\}|$ where $v = \argmax_{u\in \OPT_{\beta} \setminus (X_{t-1} \cup  \{o^*\})}\{ \frac{f(u|X_{t-1})}{c(u)} \}$.

    \noindent
   $\bullet$ \textit{Case I-2.1:} $c(X_{t-1}) \geq c(\OPT_{\beta}) - c(o^*)$\\
        We argue that, in this case, the algorithm already has admitted a $\nicefrac{1}{2}$-approximation at the end of iteration $t-1$.
        \begin{align*}
            f(F_{\omega(t)}) &\geq f(X_{t-1} \cup \{v\}) \geq f(X_{t-1})\\ &\geq \frac{c(X_{t-1})}{2(c(\OPT_{\beta})-c(o^*))}f(\OPT_{\beta}) \quad\quad\mbox{due to I.H.}\\
            &\geq \frac{1}{2} f(\OPT_{\beta}).
        \end{align*}
    \noindent
    $\bullet$ \textit{Case I-2.2:} $c(X_{t-1}) < c(\OPT_{\beta}) - c(o^*)$.
    
    We proceed by proving the following claim.
    
    \begin{claim}
        WLOG, for any iteration $t<T$,  when $c(X_{t-1}) < c(\OPT_{\beta}) - c(o^*)$ we claim that
        \begin{equation}
        \label{eq:lem-assumption}
            f(X_{t-1} \cup \{o^*\}) < \nicefrac{1}{2} f(\OPT_{\beta}).
        \end{equation}
        Otherwise, at the end of iteration $t-1$, the algorithm has already admitted a $1/2$-approximation.
        \label{claim:lem-assumption}
    \end{claim}
    \begin{proof}
        Suppose $f(X_{t-1} \cup \{o^*\}) \geq \nicefrac{1}{2} f(\OPT_{\beta})$, 
        We argue this claim in two cases: 
        
        (1) If $o^* \in X_{t-1}$, then $X_{t-1} \cup \{o^*\} = X_{t-1}$. In this case, $f(X_{t-1}) \geq \nicefrac{1}{2} f(\OPT_{\beta})$, which is guaranteed a $\nicefrac{1}{2}$-approximation. 
        
        (2) If $o^* \notin X_{t-1}$, then due to lines 11-14 of the algorithm, there exists $G'_{w(t-1)}$, which is an augmentation
        of $G_{w(t-1)}$. That is, 
        $|G'_{w(t-1)}| = |X_{t-1} \cup \{o^*\}|$ and $f(G'_{w(t-1)}) \geq f(X_{t-1} \cup \{o^*\}) \geq \nicefrac{1}{2} f(\OPT_{\beta})$. 
        Note that as per the premise of the claim, $c(X_{t-1}) < c(\OPT_{\beta}) - c(o^*)$, i.e., $c(X_{t-1}) + c(o^*)  < c(\OPT_{\beta}) \leq \beta$.  
    \end{proof}
    Going back to the proof for Case I-2 of the Lemma~\ref{lem:Xt}, 
    as per its condition: $c(X_{t-1}) < c(\OPT_\beta) - c(o^*)$. 
    By rearranging Lemma \ref{lem:greedy} with $S=X_{t-1}$, we have:
    \begin{align*}
    &f(X_{t-1} \cup \{v\})\\
    &\geq \frac{c(v)}{c(\OPT_{\beta})-c(o^*)} f(\OPT_{\beta}) + f(X_{t-1})  - \frac{c(v)}{c(\OPT_{\beta})-c(o^*)}f(X_{t-1} \cup \{o^*\})\\
        &\geq \frac{c(v)}{c(\OPT_{\beta})-c(o^*)} f(\OPT_{\beta}) + f(X_{t-1}) - \frac{c(v)}{2(c(\OPT_{\beta})-c(o^*))}f(\OPT_{\beta}) \quad \mbox{\emph{due to Claim~\ref{claim:lem-assumption}}}\\
        &= \frac{c(v)}{2(c(\OPT_{\beta})-c(o^*))}f(\OPT_{\beta}) + f(X_{t-1}) \\
        &\geq \frac{c(v)}{2(c(\OPT_{\beta})-c(o^*))}f(\OPT_{\beta}) + \frac{c(X_{t-1})}{2(c(\OPT_{\beta})-c(o^*))}f(\OPT_{\beta}) \quad \mbox{\emph{due to I.H.}}\\
        &=\frac{c(X_{t-1} \cup \{v\})}{2(c(\OPT_{\beta})-c(o^*))}f(\OPT_{\beta}).
    \end{align*}

    Now observe that
    for a specific cardinality of the parameter, the function $g(\cdot)$ is non-decreasing. Therefore, 
    \[g(X_t) \geq g(X_{t-1} \cup \{v\}) \Rightarrow \frac{f(X_t)}{c(X_t)} \geq \frac{f(X_{t-1}\cup \{v\})}{c(X_{t-1}\cup \{v\})}.\]
    
    Thus, using the above inequalities, we have:
    \begin{align*}
        f(X_t) &\geq \frac{c(X_t)}{c(X_{t-1}\cup \{v\})} f(X_{t-1}\cup \{v\})  \geq \frac{c(X_{t})}{2(c(\OPT_{\beta})-c(o^*))}f(\OPT_{\beta}).
    \end{align*}

This concludes the proof for Lemma \ref{lem:Xt}.
\end{proof}
\paragraph{Justification for \nicefrac{1}{2} - approximation.}
If at the end of some iteration $t'$, $c(X_{t'}) \geq c(\OPT_\beta)-c(o^*)$, then by Lemma \ref{lem:Xt}, \textsc{evo-SMC} admits \nicefrac{1}{2} - approximation. 

If there does not exist such $t'$, we consider the last iteration $t$ where $\omega$ is incremented. 
Similar to the proof of Case I-2 in Lemma \ref{lem:Xt}, we analyze $c(X_{t-1} \cup \{v\})$. 
Note that, $\omega$ is incremented for a good mutation where $v = \argmax_{e \in \OPT_{\beta} \setminus (X_{t-1} \cup \{o^*\})} \{ \frac{f(e|X_{t-1})}{c(e)}\}$. In that case,  we can claim
$c(X_{t-1} \cup \{v\}) \leq \beta$.
This follows from the premise for the case
\begin{align*}
 c(X_{t-1}) &< c(\OPT_\beta) - c(o^*) \leq  c(\OPT_\beta) - c(v) \leq  \beta - c(v) \\
\Rightarrow\  & c(X_{t-1} \cup \{v\}) \leq  \beta .
\end{align*}
(Recall that $o^*$ is the element in $\OPT_\beta$ with the maximal cost.) 

After such a good mutation, if $c(X_{t-1} \cup \{v\}) \geq c(\OPT_{\beta})-c(o^*)$, then the algorithm will output the optimal solution. The reason is that $\omega(t)$ can be up to $n$, i.e., the potential solution set size $|X_t|=|X_{t-1} \cup \{v\}|$ can be as large as $n$. 
Otherwise, the algorithm can keep running until $c(X_{t-1} \cup \{v\}) \geq c(\OPT_{\beta})-c(o^*)$.    

\paragraph{Bound on the Number of Iterations.}
\begin{observation}
$\omega$ is incremented with a probability at least $\frac{1}{2en^2}$ during every iteration.
\label{obs:probability}
\end{observation}
\begin{proof}
We define the event $R$ as follows: at iteration $t$, $\omega$ is incremented. Event $R$ happens when the two conditions are met as per the Definition \ref{def:w}.
The first condition is satisfied with probability $ \nicefrac{1}{2n}$ (line 3 of Algorithm \ref{alg:po}). 
Since every element of $S$ is flipped with probability $\frac{1}{n}$ independently, the probability that a single element $a$ defined in Definition \ref{def:w} was flipped and nothing else, is $\frac{1}{n}\cdot (1-\frac{1}{n})^{n-1}$. The mutation process is independent of the iteration number $t$. Therefore, 
    $\Pr[R] = \frac{1}{2n} \cdot \frac{1}{n}\cdot (1-\frac{1}{n})^{n-1} \geq \frac{1}{2en^2}.$
\end{proof}

\begin{lemma}
    In Algorithm \ref{alg:po}, define variable $Y_i$ for iteration $i\in [1,T]$ where $Y_i=1$ if $\omega(i)=\omega(i-1)+1$ and $Y_i=0$ otherwise. 
    If $T\geq \max\{4en^2K_\beta, 16en^2\log n\}$, then
    \[ \Pr\left[\sum_{i=1}^T Y_i < K_{\beta}\right] \leq \frac{1}{n}. \]
    \label{lem:time-po}
\end{lemma}
\begin{proof}
    Observation \ref{obs:probability} gives the lower bound probability of $Y_i=1$ for an iteration $i$. Therefore, $\mu = \E[Y] = \E\left[\sum_{i=1}^T Y_i\right] = T \rho \geq T\cdot \frac{1}{2en^2}$. It is immediate that $Y_i's$ are independent and take binary values with identical expectations, so we can apply Chernoff bound as presented in Lemma \ref{lem:chernoff} to claim the following inequalities.
    \begin{align*}
        \Pr\left[\sum_{i=1}^T Y_i < K_{\beta}\right] &\leq \Pr\left[\sum_{i=1}^T Y_i < T \rho/2\right]\leq e^{-T\rho/8} \leq \frac{1}{n}.
    \end{align*}
    The first inequality holds if we take $T\geq 4en^2K_{\beta}$.
    The second inequality follows from the Chernoff bound with $\eta=\nicefrac{1}{2}$. The last inequality holds because $T\geq 16en^2\log{n}$.
\end{proof}

Lemmas \ref{lem:Xt} and \ref{lem:time-po}  conclude the proof of Theorem \ref{th:po}.

\section{Stochastic Evolutionary Algorithm {\sc st-evo-SMC}}
\label{sec:bpo}
In this section, we correlate the idea of biased technique from \cite{crawford2019faster} and design a faster algorithm {\sc st-evo-SMC} for Problem \ref{pro:smk}.

\begin{algorithm}[t]
    \SetKwInOut{Input}{Input}
    \SetKwInOut{Output}{Output}
    \Input{$f:2^V \rightarrow \mathbb{R}_{\geq 0}, c:2^V  \rightarrow \mathbb{R}_{> 0}$\\ $T\in \mathbb{Z}_{> 0}, \beta \in \mathbb{R}_{> 0}, \epsilon \in (0,1], p \in [0,1]$}
    \Output{$\argmax_{X \in \{F_0,\cdots,F_{n-1},G_0, \cdots, G_{n-1},G_0', \cdots, G_{n-1}'\}}f(X)$}
    $F_j \gets \emptyset, G_j \gets \emptyset, G'_j \gets \emptyset$,  for all $j \in [0, n-1]$\\
    $\omega \gets 0, \ell \gets 1, H \gets \lceil en\log{(1/\epsilon)}\rceil$\\
    \For{$t \gets 1$ \textbf{to} $T$} {
        $S \gets \text{Random}(\{F_0,\cdots,F_{n-1},G_0, \cdots, G_{n-1}\})$\\
        \If{$\textsc{Flip-Coin}(p)=\text{heads}$} {
            $S \gets G_\omega$\\
            $\ell \gets \ell + 1$\\
            \If{$\ell~ \% ~H = 0$}{
                $\omega \gets \omega + 1$\\
            }
        }
        $S' \gets$ \Call{\Mutate}{$S$}\\
        $i \gets |S'|$\\
        \If{$c(S') \leq \beta$}{
            \If{$f(F_i) < f(S')$}{
                $F_i \gets S'$
            }
            \If{$g(G_i) < g(S')$}{
                $Q \gets G_i \cup \{v\}$ s.t.
                $v=\argmax_{e \in V\setminus G_i,c(e) \leq \beta - c(G_i)}f(e \mid G_i)$\\
                \If{$f(Q) > f(G'_i)$}{$G'_i \gets Q$}
                $G_i \gets S'$
            }
        }
    }    
    \caption{{\sc st-evo-SMC}}
    \label{alg:bpo}
\end{algorithm}

The algorithm (Algorithm \ref{alg:bpo}) has input parameters $\epsilon \in (0, 1]$ and $p \in (0,1]$ to allow users to balance the approximation guarantee and running time. The larger $\epsilon$ is, the lower the approximation the algorithm admits, and fewer iterations are required. 
One can notice that {\sc st-evo-SMC} is deducted to Algorithm \ref{alg:po} if $p=0$. On the other hand, if $p$ is very large, the algorithm becomes a sampling-based stochastic algorithm similar to \cite{mirzasoleiman2015lazier}. The given parameter $p$ controls the stochasticity probability of finding a ``good" candidate set in the first condition of Definition \ref{def:w}. In Algorithm \ref{alg:po}, the probability of randomly selecting a set for mutation is $\frac{1}{2n}$. However, in {\sc st-evo-SMC}, the probability is at least $p$. Therefore, Algorithm \ref{alg:bpo} can grow $\omega$ rapidly and results in smaller amount of iterations.

\begin{theorem}
    With a monotone submodular function $f$, modular cost function $c$, cost constraint $\beta$, the maximal seed set size $K_\beta$, error threshold $\epsilon \in (0,1]$ and stochasticity probability $p \in (0,1]$, let $\OPT_{\beta}=\argmax_{X\subseteq V, c(X)\leq \beta}\{f(X)\}$. After $T\geq 2enK_{\beta}\ln{(1/\epsilon)}/p $ iterations, Algorithm \ref{alg:bpo} outputs $X$ and admits $\frac{1}{2}$- approximation with probability $1-\epsilon$.
\label{th:bpo}
\end{theorem}

\noindent
{\em Proof Sketch:} The idea is to bond the increment of $\omega$ (line 9) with a good mutation. If $\omega$ is incremented, then $G_\omega$ at line 6 has been selected for mutation for at least $H$ time, and it is very likely that there ever exists at least one good mutation during those mutations.  
If we define $E_t$: $\omega$ was incremented at the end of iteration $t$; $F_t$: Given $E_t$, there exists at least one good mutation during some iteration(s) since the last iteration $\omega$ was incremented. We have Claim \ref{claim:Ft} showing that $F_t$ happens with probability $1-\epsilon$ given that $E_t$ happens.
\begin{claim}
\label{claim:Ft}
    $\Pr[F_t \mid E_t] \geq 1-\epsilon$.
\end{claim}
Analogous to the proof of Theorem \ref{th:po}, we need the following lemma to guarantee a $\nicefrac{1}{2}$-approximation ratio.
\begin{lemma}
    For every iteration $t$ such that $1 \leq t \leq T$, let $X_t$ be the set with size $\omega(t)$ at the end of iteration $t$, then with probability $1-\epsilon$, we have:
    \begin{equation}
        f(X_t) \geq \frac{c(X_t)}{2(c(\OPT_{\beta})-c(o^*))}f(\OPT_{\beta}) .
        \label{eq:bpo-lemma}
    \end{equation} 
    \label{lem:Xt-bpo}
\end{lemma}
\noindent
{\em Proof Sketch:} The proof uses strong induction proof methods. The base cases include iterations with $\omega=0$ at time and the first iteration with $\omega=1$. Conditioned on $\neg E_t$, Lemma \ref{lem:Xt} can be verified; when $F_t$ occurs, applying Lemma \ref{lem:greedy} will finish the proof. For inductive steps, the inductive hypothesis is: Eq. (\ref{eq:bpo-lemma}) is true for all iterations $i$ such that $1 \leq i<t$. The proof follows similar to Case I-2 proof in Section \ref{sec:po}.

The following lemma ensures that $\omega$ has been incremented for $K_\beta$ times after $T$ iterations with probability $1-\epsilon$.
\begin{lemma}
    In Algorithm \ref{alg:bpo}, define $0-1$ variables  $Y_i$ for iteration $i\in [1,T]$ where $T \geq 2enK_{\beta}\ln{(1/\epsilon)}/p$. $Y_i=1$ if $\ell$ is incremented at line 7 of Algorithm \ref{alg:bpo}, $0$ otherwise.
    Then,
    \[ \Pr\left[\sum_{i=1}^T Y_i < HK_{\beta}\right] \leq \epsilon. \]
\end{lemma}

\section{Experiments}
\label{sec:experiements}
{\sc evo-SMC} and {\sc st-evo-SMC} are evaluated on influence maximization, directed vertex cover, and sensor placement with costs. 
We compare our algorithms with the evolutionary EAMC \cite{bian2020efficient} and the deterministic algorithm Greedy+Max \cite{yaroslavtsev2020bring} (\nicefrac{1}{2}-approximation and runs in $O(nK_\beta)$). 
We implement our algorithms and baselines in C++ \footnote{\texttt{https://github.com/yz24/evo-SMC}.}.
We run our algorithms and EAMC 20 times and report the medians. 
% Additional results and more technical details are deferred to the appendix.

\paragraph{Implementation Accelerations.} The mutation procedures in algorithms {\sc evo-SMC}, {\sc st-evo-SMC} and EAMC are performed independently at every iteration. Therefore, the set after mutation $S'$ can stay the same with non-trivial probability or be equal to some set that has already been considered. Thus, to avoid repeated evaluations of the objective $f$ and $g$ over the same set, we maintain a bloom filter \cite{bloom1970space} and pre-check the incoming set $S'$. 
% The implementation details are deferred to the appendix.
In our implementations and results, the number of evaluations of $f$ is incremented only if an evaluation of $f(S')$ is actually performed.

\subsection{Applications and Experimental Settings}
\paragraph{Influence Maximization with Costs.}
Influence maximization in a social network $G=(V, E)$ seeks to maximize $\E[IC(X)]$, the expected number of influenced users influenced by the propagation of information from a subset of users (seed set) $X \subseteq V$. The Independent Cascade (IC) model estimates $\E[IC(X)]$ with propagation probabilities $p(u,v)$ \cite{kempe2003maximizing}. In real-life scenarios, there is a total budget $\beta$, and each node $v \in V$ has a cost $c(v)$ (can be viewed as an incentive) in the propagation process. Hence, the cost-constrained influence maximization problem can be formulated as: finding $X$ such that $\argmax_{X\subseteq V, c(X) \leq \beta} \E[IC(X)]$. We define a linear cost function $c:V \rightarrow \mathbb{R}^+$ proportional to the out-degree $out(v)$ \cite{jin2021unconstrained}, where $c(v) = \lambda \cdot out(v)^\gamma$ with free parameters $\gamma$ and $\lambda$. If $out(v)=0$, $c(v)$ is set to $1$. In our experiments, we use Facebook \cite{leskovec2012learning} and Film-Trust networks \cite{kunegis2013konect}.
We run the algorithms with fixed budget $\beta=20$ and generate node costs with $\lambda=1.2, \gamma=1.5$.

% %%%%%%%%%%%%%%%%%%%%%%%%%%%%%%%%%%%%%%%%%%%%
\begin{figure}[t]
\centering
\subfigure[Film-Trust, $p=0.5$]{
\centering
\includegraphics[width=0.42\textwidth]{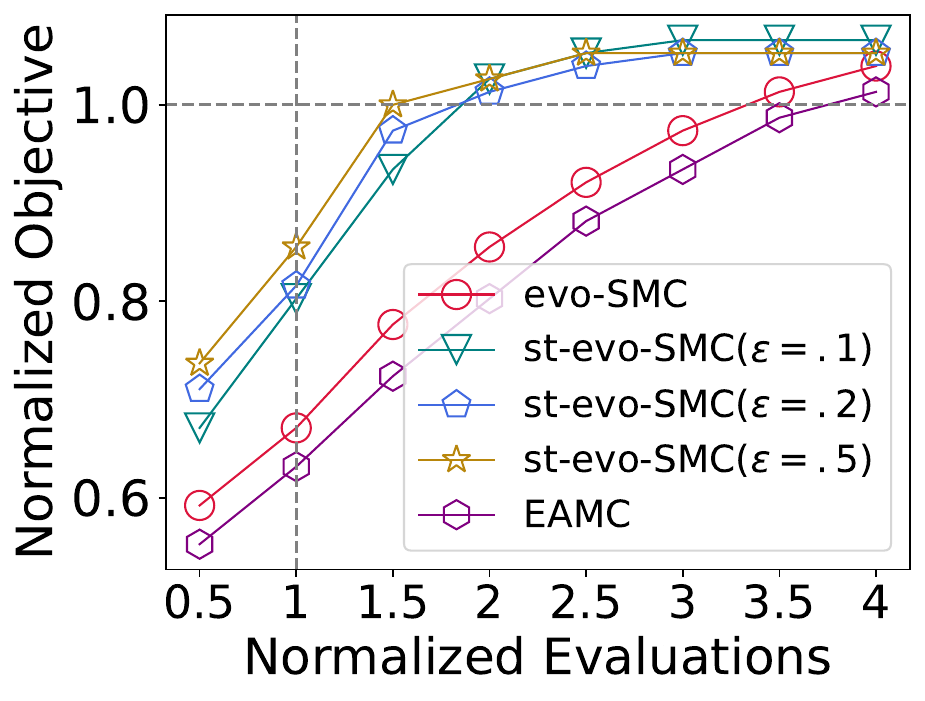}
% \caption{}
\label{fig:film-e}
}
\subfigure[Facebook, $p=0.5$]{
\centering
\includegraphics[width=0.42\textwidth]{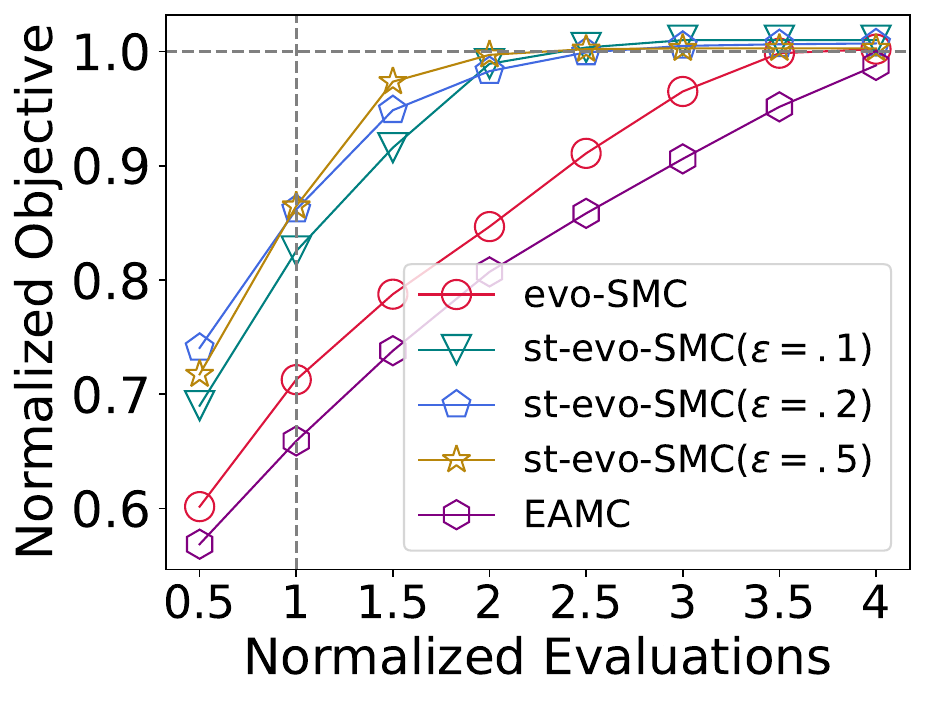}
% \caption{}
\label{fig:fb-e}
}
\caption{Comparisons on Film-Trust and Facebook networks.}
\end{figure}
% %%%%%%%%%%%%%%%%%%%%%%%%%%%%%%%%%%%%%%%%%%%%
\begin{figure}[t]
\centering
\subfigure[Eu-Email; $q=5$, $p=0.5$]{
\centering
\includegraphics[width=0.42\textwidth]{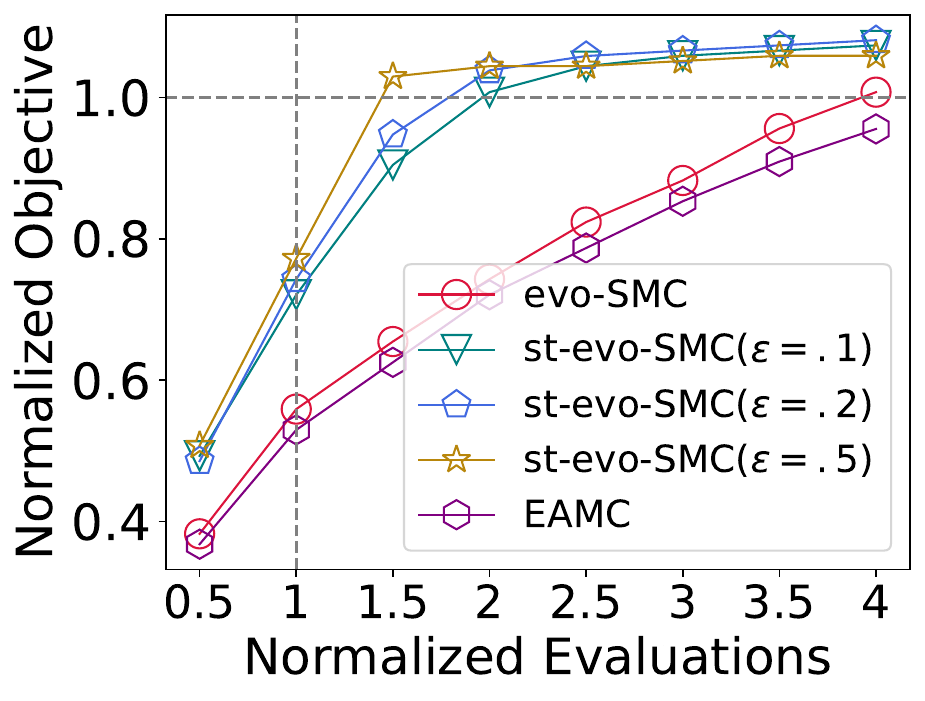}
% \caption{}
\label{fig:email-q5-e}
}
\subfigure[Protein; $q=5$, $p=0.5$]{
\centering
\includegraphics[width=0.42\textwidth]{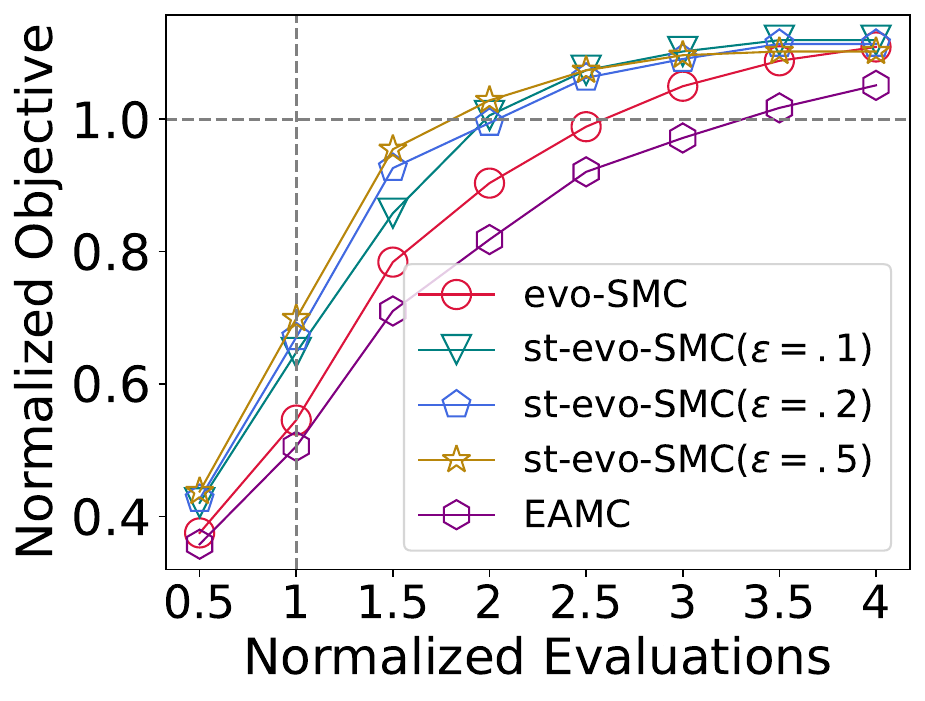}
% \caption{}
\label{fig:protein-q5-e}
}
\caption{Comparisons on Eu-Email and Protein network with cost penalty $q=5$.}
\end{figure}

% %%%%%%%%%%%%%%%%%%%%%%%%%%%%%%%%%%%%%%%%%%%%
\begin{figure}[t]
\centering
\subfigure[Beijing; $p=0.5$]{
\centering
\includegraphics[width=0.42\textwidth]{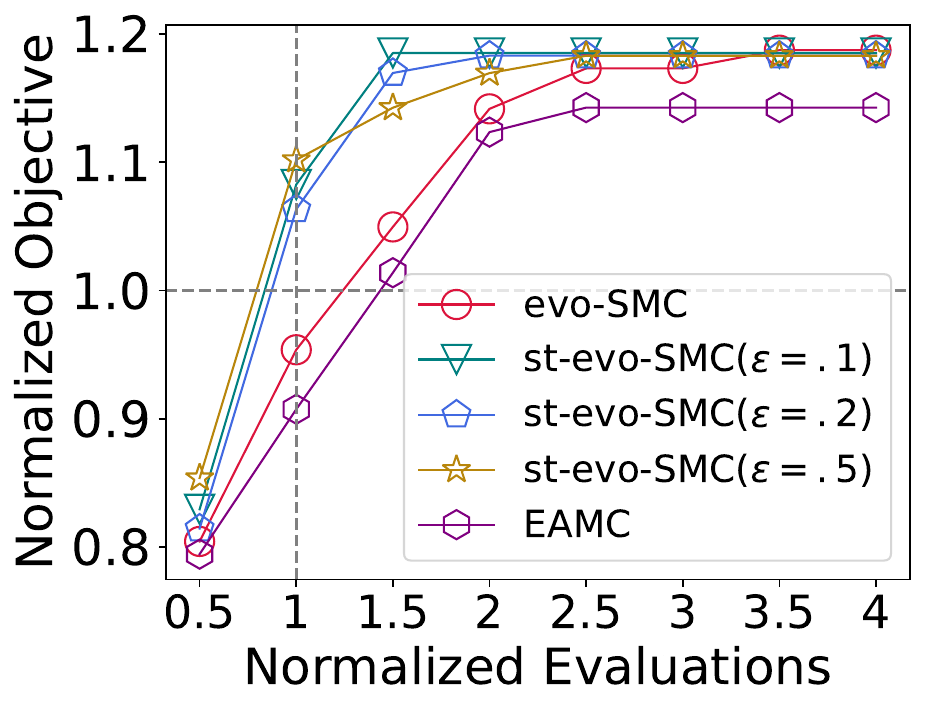}
% \caption{}
\label{fig:bj-p}
}
\subfigure[Beijing; $\epsilon=0.1$]{
\centering
\includegraphics[width=0.42\textwidth]{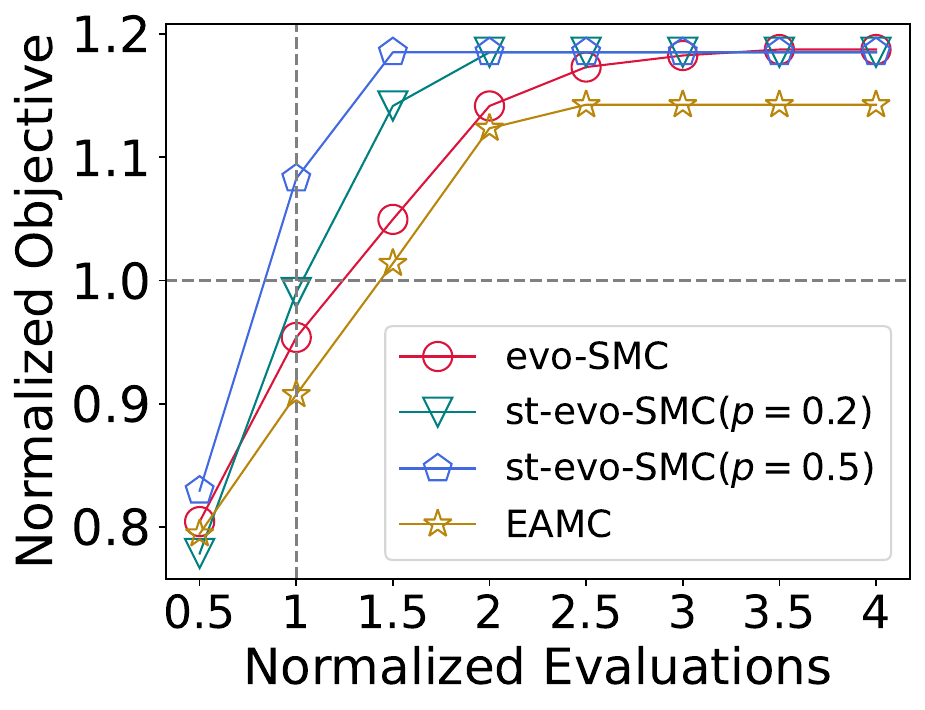}
% \caption{}
\label{fig:bj-eps}
}
\caption{Comparisons on Beijing data with various $\epsilon$ and $p$.}
\end{figure}
% %%%%%%%%%%%%%%%%%%%%%%%%%%%%%%%%%%%%%%%%%%%%

\paragraph{Directed Vertex Cover with Costs.}
Let $G=(V, E)$ be a directed graph and $w: 2^V \rightarrow \mathbb{R}^{\geq 0}$ be a modular weight function on a subset of vertices. For a vertex set $S \subseteq$ $V$, let $N(S)$ denote the set of vertices which are pointed to by $S$, formally,
$N(S) =\{v \in V \mid(u, v) \in E \wedge u \in S\}$. The weighted directed vertex cover function is $f(S)=\sum_{u \in N(S) \cup S} w(u)$, which is monotone submodular. We also assume that each vertex $v \in$ $V$ has an associated non-negative modular cost function $c(v)$ \cite{harshaw2019submodular} defined by $c(v)=1+\max \{d(v)-q, 0\}$, where $d(v)$ is the out-degree of vertex $v$ and the non-negative integer $q$ is the cost penalty. 
% When $q$ is large, the vertex costs are small. 
The objective of this task is to find a subset $S$ such that $\argmax_{S\subseteq V, c(S) \leq \beta}\sum_{u \in N(S) \cup S} w(u)$. We use Protein network \cite{stelzl2005human} and Eu-Email network~\cite{leskovec2007graph} in this application. 
We assign each node a weight of 1 and generate costs as mentioned above. We report results of $\beta=30$ and cost penalty $q=5$ for both networks.

\paragraph{Sensor Placement with Costs.}
We use a real-world air quality data (light and temperature measures) \cite{zheng2013u} \cite{bian2020efficient} collected from 36 monitoring stations in Beijing. We calculate the entropy of a sensor placement using the observed frequency. Each sensor is assigned a positive cost from the normal distribution $\mathcal{N}\left(0,1\right)$. The goal is to maximize the submodular Entropy function with a budget $\beta$. In this experiment, we set $\beta = 10$.

\noindent
\subsection{Results} 
The comparison results on influence maximization (Film-Trust and Facebook), vertex cover (Protein and Eu-Email) and sensor placement (Beijing) are illustrated in Figures \ref{fig:film-e}, \ref{fig:fb-e}, Figures~\ref{fig:email-q5-e}, \ref{fig:protein-q5-e} and Figures \ref{fig:bj-p}, \ref{fig:bj-eps}, respectively. 
Note that the objective values and runtime are normalized by the objective value and number of oracle calls made ($O(nK_\beta)$) by the Greedy+Max algorithm. The grey lines in the figures are $1$'s, corresponding to the objective value and oracle evaluations of Greedy+Max. 

Due to the accelerations with bloom filters and pre-check step, our algorithm {\sc evo-SMC} is as efficient as EAMC. Moreover, the stochastic version algorithm {\sc st-evo-SMC} consistently performs better than EAMC and outperforms Greedy+Max after making approximately two times oracle evaluations of Greedy+Max (exceeding the grey horizontal lines). As expected, the objective values of {\sc evo-SMC} grow slowly compared to {\sc st-evo-SMC}. The reason is that {\sc evo-SMC} randomly selects a set to mutate while {\sc st-evo-SMC} performs the selection process with a purpose (with a stochasticity probability $p$). We can also observe that when $\epsilon=0.1$, {\sc st-evo-SMC} produces higher-quality results than with larger error thresholds, which is consistent with the theoretical guarantees. In Figure \ref{fig:bj-p}, we plot the results with various $p$ choices. With a larger $p=0.5$, {\sc st-evo-SMC} has a faster initial increase compared with {\sc st-evo-SMC}($p=0.2$).

During the experiments, we observe that the valuation of cost of partial solutions accumulates to $\beta$ quickly and then stay close to $\beta$. This means there were a very small portion or no elements to evaluate to find the augmented element (line 19 of Algorithm \ref{alg:bpo}), which is a reason why {\sc evo-SMC} is as efficient as EAMC.   

\section{Conclusions}
\label{sec:conclusion}
In this paper, we proposed novel evolutionary frameworks for submodular maximization with cost constraints. Our algorithms achieve competitive approximation guarantees compared to the state-of-the-art evolutionary methods. Empirical studies demonstrate that our algorithms are also efficient. Future work would be designing faster algorithms that can efficiently adapt to the dynamic cost constraint settings. Unlike EAMC \cite{bian2020efficient}, the surrogate function defined in our paper is independent of the budget $\beta$. Therefore, it is promising to investigate how our framework can deal with budget changes in dynamic cost settings. 

\section*{Acknowledgements}
Part of Pavan's work is supported by NSF grant 2130536

\appendix
\newtheorem{manualtheoreminner}{Theorem}
\newenvironment{manualtheorem}[1]{%
  \renewcommand\themanualtheoreminner{#1}%
  \manualtheoreminner
}{\endmanualtheoreminner}

\newtheorem{manuallemmainner}{Lemma}
\newenvironment{manuallemma}[1]{%
  \renewcommand\themanuallemmainner{#1}%
  \manuallemmainner
}{\endmanuallemmainner}

\newtheorem{manualclaiminner}{Claim}
\newenvironment{manualclaim}[1]{%
  \renewcommand\themanualclaiminner{#1}%
  \manualclaiminner
}{\endmanualclaiminner}

\section{Omitted Proofs in Section \ref{sec:bpo}}
\label{sec:app-BPO}

In this section, we present the omitted proofs for Algorithm \ref{alg:bpo} in Section \ref{sec:bpo}. We also provide the \Mutate procedure in Algorithm \ref{alg:mutate} used as a subroutine in algorithms {\sc evo-SMC} and {\sc st-evo-SMC}.

% \begin{algorithm}
%     \SetKwInOut{Input}{Input}
%     \SetKwInOut{Output}{Output}
%     \Input{$f:2^V \rightarrow \mathbb{R}_{\geq 0}, c:2^V  \rightarrow \mathbb{R}_{> 0}$\\ $T\in \mathbb{Z}_{> 0}, \beta \in \mathbb{R}_{> 0}, \epsilon \in (0,1], p \in [0,1]$}
%     \Output{$\argmax_{X \in \{F_0,\cdots,F_{n-1},G_0, \cdots, G_{n-1},G_0', \cdots, G_{n-1}'\}}f(X)$}
%     $F_j \gets \emptyset, G_j \gets \emptyset$ for all $j \in [0, n-1]$\\
%     $\omega \gets 0, \ell \gets 1, H \gets \lceil en\log{(1/\epsilon)}\rceil$\\
%     \For{$t \gets 1$ \textbf{to} $T$} {
%         $S \gets \text{Random}(\{F_0,\cdots,F_{n-1},G_0, \cdots, G_{n-1}\})$\\
%         \If{$\textsc{Flip-Coin}(p)=\text{heads}$} {
%             $S \gets \argmax_{X \in \{G_0,\cdots,G_\omega\}} g(X)$\\
%             $\ell \gets \ell + 1$\\
%             \If{$\ell~ \% ~H = 0$}{
%                 $\omega \gets \omega + 1$\\
%             }
%         }
%         $S' \gets$ \Call{\Mutate}{$S$}\\
%         $i \gets |S'|$\\
%         \If{$c(S') \leq \beta$}{
%             \If{$f(F_i) < f(S')$}{
%                 $F_i \gets S'$
%             }
%             \If{$g(G_i) < g(S')$}{
%                 $G_i \gets S'$\\
%                 $G_i' \gets G_i \cup \{v\}$, $v=\argmax_{e \in V\setminus G_i}f(v \mid G_i)$
%             }
%         }
        
%     }    
%     \caption{\textsc{c-BPO}}
%     % \label{alg:bpo}
% \end{algorithm}

\begin{algorithm}[ht]
    \SetKwInOut{Input}{Input}
    \SetKwInOut{Output}{Output}
    \Input{$S$}
    \Output{$S'$}
    $S' \gets S$\\
    \For{$x \in V$} {
        
        \eIf{$x \in S$} {
            $S' \gets S' \setminus\{x\}$ with probability $1/n$
        }
        {
            $S' \gets S'\cup \{x\}$ with probability $1/n$
        }
    }
    \caption{\Mutate}
    \label{alg:mutate}
\end{algorithm}

% \begin{definition}
% For any iteration $t \in \mathbb{Z}_{\geq 1}$,
%     \[ X_t = \argmax_{X \in \{G_0, \cdots, G_{\omega(t)}\}} g(X)\] 
%     \label{def:X_t-bpo}
% \end{definition}
\begin{manualtheorem}{2}
    With a monotone submodular function $f$, modular cost function $c$, cost constraint $\beta$, the maximal seed set size $K_\beta$, error threshold $\epsilon \in (0,1]$ and stochasticity probability $p \in (0,1]$, let $\OPT_{\beta}=\argmax_{X\subseteq V, c(X)\leq \beta}\{f(X)\}$. After $T\geq 2enK_{\beta}\ln{(1/\epsilon)}/p $ iterations, Algorithm \ref{alg:bpo} outputs $X$ and gives $\frac{1}{2}$ approximation with probability $1-\epsilon$.
    % \[\E[f(X)] \geq (\frac{1}{2} - \epsilon)f(\OPT_{\beta})\]
\label{app:th:bpo}
\end{manualtheorem}

\begin{proof}
We start with a definition $\sigma(t)$, which characterizes the smallest iteration $i\in[0,t]$ with $\omega(i) = \omega(t-1)$. In other words, $\sigma(t)$ keeps track of the last time $\omega$ was incremented.

\begin{definition}
For every iteration $t \in \mathbb{Z}^{+}$, define
    \[ \sigma(t) = \min\{i\in \{ 0, \cdots, t \}: \omega(i) = \omega(t-1) \}.\] 
    \label{def:sigma_t}
\end{definition}

At the beginning of the algorithm, $\omega=0$, and $\omega$ is incremented if $\ell$ has been incremented for sufficient times (lines 5-11 of Algorithm \ref{alg:bpo}). Table \ref{t:sigma} lists an example of $\sigma(t)$ and the interval $(\sigma(t),t]$.  In this example scenario, the underlined valuations in the first column are the iterations when the corresponding $\omega$ is incremented.  

\begin{table*}[ht]
    \centering
    \caption{$\sigma(t)$ and intervals for the example}
    \begin{tabular}{ |p{2cm}|p{2cm}|p{2cm}|p{2cm}|  }
    \hline
    
    $t$ & $\omega(t)$ & $\sigma(t)$ & $(\sigma(t),t]$ \\
    \hline
    1 & 0 & 0 & $(1,1]$ \\
    {2} & 0 & 0 & $(1,2]$ \\
    \underline{3} & 1 & 0 & $(1,3]$ \\
    {4} & 1 & 3 & $(3,4]$ \\
    {5} & 1 & 3 & $(3,5]$ \\
    \underline{6} & 2 & 3 & $(3,6]$   \\
    {7} & 2 & 6 & $(6,7]$ \\
    $\vdots$ & $\vdots$ & $\vdots$ & $\vdots$\\
    {19} & 2 & 6 & $(6,19]$\\
    \underline{20} & 3 & 6 & $(6,20]$\\
    $\vdots$ & $\vdots$ & $\vdots$ & $\vdots$\\
    \hline
    \end{tabular}
    \label{t:sigma}
\end{table*}

At iteration $t$, we define events $E_t$ and $F_t$ as follows:
\begin{itemize}
    \item $E_t$: At the end of iteration $t$, $\omega(t)=\omega(t-1)+1$, i.e., $\ell$ has been incremented for $H=\lceil en\log{(1/\epsilon)}\rceil$ times. (At both iterations $\sigma(t)$ and $t$, $\omega$'s are incremented and during every iteration $r\in (\sigma(t), t)$, $\omega$'s stay the same.)  
    
    \item $F_t$: $F_t$ is a sub-event of $E_t$. $F_t$ occurs if $\omega$ is incremented at iteration $t$ {\em and} during at least one iteration $r\in (\sigma(t), t]$, a good mutation happened. I.e., there is a single element $v$ added to $X_{r-1}$, and the element is exactly $v = \argmax_{e\in \OPT_{\beta} \setminus (X_{r-1} \cup \{o^*\})} \frac{f(e|X_{r-1})}{c(e)}$.
\end{itemize}

In the example of Table \ref{t:sigma}, at iterations $t=3, 6, 20$, at both ends of the interval $(\sigma(t),t]$, the corresponding $\omega({\sigma(t)})$ and $\omega(t)$ are incremented. By definitions of events $E_t$ and $F_t$, the underlined iterations in the table are event $E_t$ happened and others are $\neg E_t$.

The next claim correlates the increment of $\omega$ with a good mutation at the end of an iteration $t$.
% \noindent
% \medskip
% \textbf{Remark: } $\Pr[E_t] = \Pr[F_t \mid E_t] + \Pr[\neg F_t \mid E_t]$

\begin{manualclaim}{2}
For every iteration $t \in [1, T]$,
    \[\Pr[F_t \mid E_t] \geq 1-\epsilon.\]
    \label{claim:Ft-app}
\end{manualclaim}
\begin{proof}
    The probability that the \Mutate procedure flips the membership of a specific element, nothing else, is 
    \[ \rho = \frac{1}{n}\cdot (1-\frac{1}{n})^{n-1}\geq \frac{1}{en}.\]
    Given $E_t$ happens, $\ell$ has been incremented for $H$ times during iterations $(\sigma(t), t]$. Among the $H$ mutations, \textit{at least one} good mutation exists will result in $F_t$. Therefore, the probability of $\Pr[F_t \mid E_t]$ is   
    \begin{align*}
        \Pr[F_t \mid E_t] &\geq 1 - (1-\rho)^H \\
        &\geq 1-(1-\frac{1}{en})^{en\log{(1/\epsilon)}}\\
        &\geq 1-\epsilon.
    \end{align*}
    
\end{proof}

\begin{manuallemma}{5}
    For every iteration $t$ such that $1 \leq t \leq T$, let $X_t$ be the set with size $\omega(t)$ at the end of iteration $t$, then with probability $1-\epsilon$, we have:
    \begin{equation}
        f(X_t) \geq \frac{c(X_t)}{2(c(\OPT_{\beta})-c(o^*))}f(\OPT_{\beta}). 
        \label{app:eq:bpo-lemma}
    \end{equation}
    \label{app:lem:Xt-bpo}
\end{manuallemma}

\begin{proof}
We prove this lemma by strong induction. 

\medskip
\noindent
\textbf{Base Cases: }
The base cases of the strong induction include all iterations $t \in [1,d]$ where $d = \min\{t \in \{1, \cdots, T\}: \omega(t) =1\}$. It is clear that during iterations $t \in [1,d)$, $\omega(t)=0$, which is $\neg E_t$. At iteration $t=d$, $E_t$ happens. So, we have two sub-cases B-1 and B-2.
\begin{itemize}
    \item \textbf{Case B-1}. For every iteration $t \in [1,d)$, $\omega(t)=0$, event $E_t$ did not happen. $\omega$ was not incremented, then we have $\omega(t)=\omega(0)=0$, $|X_t|=|X_0|=0$. We can verify the lemma holds by:
        \[f(X_t) = f(X_0)= 0 \geq \frac{c(\emptyset)}{2(c(\OPT_{\beta})-c(o^*))}f(\OPT_{\beta})=0.\]
        
    Thus, this lemma holds for the any iteration $t \in [1,d)$.

    \item \textbf{Case B-2} At iteration $t=d$, $\omega$ was incremented ($E_t$ happened), so $|X_d|=\omega(d)=\omega(d-1)+1=1$. By Claim \ref{claim:Ft-app}, $F_t$ also happened with probability $1-\epsilon$.
    
    Then conditioned on $F_t$,
    there exists at least one good mutation at some iteration $r \in [1,d]$. Assume that the good element in Definition \ref{def:w} is $v$. By the monotonicity of function $g$ on the same size sets guaranteed by the algorithm, we have $g(X_d) \geq g(\{v\})$.

     Lemma \ref{lem:greedy} only requires the second condition of good mutation (Def. \ref{def:w}) happening on every given set $S$, therefore, this lemma also applies to Algorithm \ref{alg:bpo} conditioned on $F_t$. Applying Lemma \ref{lem:greedy}, we obtain
    \begin{align*}
        f(\{v\}) - f(\emptyset) 
        &\geq \frac{c(v)}{c(\OPT_{\beta})-c(o^*)}(f(\OPT_{\beta})-f(\{o^*\})) \quad\quad \mbox{\em due to Lem. \ref{lem:greedy} with $S=X_0=\emptyset$}\\
        & \geq \frac{c(v)}{2(c(\OPT_{\beta})-c(o^*))}f(\OPT_{\beta}). \quad\quad \mbox{\em assume $f(\{o^*\}) < \frac{1}{2}f(\OPT_{\beta})$ }
    \end{align*}
    The last inequality holds due to the assumption $f(\{o^*\}) < \frac{1}{2}f(\OPT_{\beta})$. This assumption is safe because if $f(\{o^*\}) \geq \frac{1}{2}f(\OPT_{\beta})$, Algorithm \ref{alg:bpo} already admits $\nicefrac{1}{2}$ approximation for line 19 of the algorithm.
   
    Then by the definition of $g$, 
    \[f(X_d) \geq \frac{c(X_d)}{c(v)} f(\{v\}) \geq \frac{c(X_d)}{2(c(\OPT_{\beta})-c(o^*))}f(\OPT_{\beta}).\]

    % So, by the total expectation,
    % \begin{align*}
    %     \E[f(X_d)] &= \Pr[F_t] \E[f(X_d) \mid F_t] + \Pr[\neg F_t] \E[f(X_d) \mid \neg F_t]\\
    %     &\geq \Pr[F_t] \E[f(X_d) \mid F_t]\\
    %     &\geq (1- \epsilon)\frac{c(X_d)}{2(c(\OPT_{\beta})-c(o^*))}f(\OPT_{\beta})
    % \end{align*}
Therefore, this lemma holds for iteration $t=d$.
\end{itemize}     

The proofs of Case B-1 ($\neg E_t$) and Case B-2 ($F_t$) of indicate that Lemma \ref{app:lem:Xt-bpo} holds for the base cases. The probability combined is,
\begin{align*}
    \Pr[\neg E_t] + \Pr[F_t] &= \Pr[\neg E_t] + \left(\Pr[E_t] - \Pr[\neg F_t \mid E_t]\right)\\
    &= 1-\Pr[\neg F_t \mid E_t] \\
    &= \Pr[F_t \mid E_t] \\
    &\geq 1-\epsilon. \quad\quad\quad\quad \mbox{\em follows from Claim \ref{claim:Ft-app} } 
\end{align*}

Thus, for the base cases, Lemma \ref{lem:Xt-bpo} holds with probability at least $1-\epsilon$.

\medskip
\noindent        
\textbf{Induction Steps: } Next, we prove for iterations $t\in (d, T]$ with $\omega(t) \geq 1$. We have two cases -- Event $E_t$ does not occur ($\neg E_t$) and event $E_t$ occurs ($E_t$).

We consider a fixed iteration $t$. Different from the analysis of Algorithm \ref{alg:po}, our strong induction hypothesis(I.H.) here is
\begin{equation}
    \label{eq:strong-IH}
    \forall ~r \in [1, t):~ f(X_r) \geq \frac{c(X_r)}{2(c(\OPT_{\beta})-c(o^*))}.
\end{equation}

\begin{itemize}
    \item \textbf{Case I-1: Event $E_t$ does not occur ($\neg E_t$)}.\\
    Event $E_t$ does not occur implies that during every iteration $r \in (\sigma(t), t]$, $\omega$ never incremented and $\omega(t)=\omega(r)$. Therefore, $|X_t|=|X_r|$ and the algorithm will compare their $f$ and $g$ valuations.

    At iteration $t$, the algorithm selected $X_{t-1}$ for mutation with a probability at least $p$ (line 6 of Algorithm \ref{alg:bpo}). By Definition \ref{def:g} and conditioned on $\neg E_t$, we have
    \[
        g(X_t) \geq g(X_{t-1}) \Rightarrow  \frac{f(X_t)}{c(X_t)} \geq \frac{f(X_{t-1})}{c(X_{t-1})} .
    \]
    Rearranging the above inequality, we have
    \begin{align*}
        f(X_t) &\geq \frac{c(X_t)}{c(X_{t-1})}f(X_{t-1})\\
        &\geq \frac{c(X_t)}{c(X_{t-1})}\cdot \frac{c(X_{t-1})}{2(c(\OPT_{\beta})-c(o^*))}f(\OPT_{\beta}) \quad\quad\mbox{\em due to I.H. (Eq. (\ref{eq:strong-IH}))}\\
        &=\frac{c(X_t)}{2(c(\OPT_{\beta})-c(o^*))}f(\OPT_{\beta}).
    \end{align*}
    Therefore, condition on $\neg E_t$, we have
    \begin{equation}
    \label{eq:X_t-not_Et}
        f(X_t) \geq \frac{c(X_t)}{2(c(\OPT_{\beta})-c(o^*))}f(\OPT_{\beta}).
    \end{equation}

    \item \textbf{Case I-2: Even $E_t$ occurs.}\\
    If Event $F_t$ occurs (based on $E_t$) with probability at least $1-\epsilon$, there is at least one good mutation during some iteration $r \in (\sigma(t), t]$. Suppose that a good mutation happened at iteration $r$. $X_{r-1}$ was selected for mutation and element $v$ was added (defined in the second condition of Definition \ref{def:w}). We have $\omega(t) = \omega(r)+1$ and $|X_t|=|X_{r-1} \cup \{v\}|$. The algorithm will compare the $f$ and $g$ valuations of $X_t$ and $X_{r-1} \cup \{v\}$.

    Conditioned on $F_t$ and by the definition of $g$, we obtain the following inequality:
    \begin{equation}
        g(X_t) \geq g(X_{r-1} \cup \{v\})  
        \\ 
        \Rightarrow \frac{f(X_t)}{c(X_t)} \geq \frac{f(X_{r-1} \cup \{v\}) }{c(X_{r-1}\cup \{v\}) }.
        \label{eq:g-Ft}
    \end{equation}

    Next we derive a lower bound on $f(X_{r-1} \cup \{v\})$.
    \begin{align}
         &f(X_{r-1} \cup \{v\})\nonumber \\
         &\geq \frac{c(v)}{c(\OPT_{\beta})-c(o^*)} \left(f(\OPT_{\beta})-f(X_{r-1} \cup \{o^*\}) \right)+f(X_{r-1}) \nonumber\\
        &= \frac{c(v)}{c(\OPT_{\beta})-c(o^*)} f(\OPT_{\beta}) + f(X_{r-1}) - \frac{c(v)}{c(\OPT_{\beta})-c(o^*)}f(X_{r-1} \cup \{o^*\}) \quad\quad \mbox{\em due to Lemma \ref{lem:greedy}} \nonumber\\
        &\geq \frac{c(v)}{c(\OPT_{\beta})-c(o^*)} f(\OPT_{\beta}) + f(X_{r-1}) - \frac{c(v)}{2(c(\OPT_{\beta})-c(o^*))}f(\OPT_{\beta}) \quad\quad \mbox{\em follows from Eq. (\ref{eq:lem-assumption})} \nonumber\\
        &=\frac{c(v)}{2(c(\OPT_{\beta})-c(o^*))}f(\OPT_{\beta}) + f(X_{r-1}) \nonumber\\
        &\geq \frac{c(v)}{2(c(\OPT_{\beta})-c(o^*))}f(\OPT_{\beta}) + \frac{c(X_{r-1})}{2(c(\OPT_{\beta})-c(o^*))}f(\OPT_{\beta})\quad\quad \mbox{\em due to I.H. (Eq. \ref{eq:strong-IH})}\smallskip \nonumber\\
        &\geq \frac{c(X_{r-1} \cup \{v\})}{2(c(\OPT_{\beta})-c(o^*))}f(\OPT_{\beta}). \label{eq:X_t-Et-2}
    \end{align}
    Therefore, with Eq. (\ref{eq:g-Ft}) and Eq. (\ref{eq:X_t-Et-2}), we have
    \[f(X_t) \geq \frac{c(X_t)}{c(X_{r-1} \cup \{v\})} f(X_{r-1} \cup \{v\}) \geq \frac{c(X_{t})}{2(c(\OPT_{\beta})-c(o^*))}f(\OPT_{\beta}). \]
\end{itemize}
Analogous to the analysis of the base cases, Lemma \ref{lem:Xt-bpo} holds with probability at least $1-\epsilon$ for the induction steps.
\end{proof}

Subsequently, we bound the number of total iterations $T$ of the algorithm.

\begin{manuallemma}{6}
    In Algorithm \ref{alg:bpo}, define binary random variables  $Y_i$ for iteration $i\in [1,T]$ where $T \geq 2enK_{\beta}\ln{(1/\epsilon)}/p$ such that
    \[ Y_i =
  \begin{cases}
    ~1       & \quad \ell~\text{is incremented at iteration }i\\
    ~0  & \quad \text{otherwise.}
  \end{cases}
    \]    
    % $Y_i=1$ if $\ell$ is incremented at line 7 of Algorithm \ref{alg:bpo}, $0$ otherwise.
    Then, we have
    \[ \Pr\left[\sum_{i=1}^T Y_i < HK_{\beta}\right] \leq \epsilon. \]
\end{manuallemma}
\begin{proof}
     The lower bound probability of $Y_i=1$ for an iteration $i$ is $p$, which is determined by line 5 of Algorithm \ref{alg:bpo}. Therefore, $\mu = \E[Y] = \E[\sum_{i=1}^n Y_i] = Tp$. It's trivial that $Y_i's$ are independent and take binary values \{0,1\} with identical expectations, so we can apply Chernoff bound (Lemma \ref{lem:chernoff}) to prove that after $T$ iterations, $\ell$ has been incremented for $HK_\beta$ times ($\omega$ has been incremented to $K_\beta$) with probability at least $1-\epsilon$.
    
    \begin{align*}
        \Pr\left[\sum_{i=1}^T Y_i < HK_{\beta}\right] &\leq \Pr\left[\sum_{i=1}^T Y_i < enK_{\beta}\ln{(1/\epsilon)}\right]\\
        &\leq \Pr\left[\sum_{i=1}^T Y_i < T p/2\right]\\
        &\leq e^{-Tp/8}\\
        &\leq \epsilon.
    \end{align*}

    The second inequality holds if we take $T\geq 2enK_{\beta}\ln{(1/\epsilon)}/p$. The third inequality is by applying Chernoff bound Lemma \ref{lem:chernoff} with $\eta=1/2$. The last inequality holds if we take $T\geq 8\log{(1/\epsilon)}/p$.
    
\end{proof}

The $1/2$ constant approximation can be justified using the idea in Section \ref{sec:po}.

\end{proof}
\section{Additional Experiments and Justifications}
\label{app:results}
In this section, we provide the experiment implementation details omitted in Section \ref{sec:experiements}. We also present extensive experimental results and analysis. Some claims in Section \ref{sec:experiements} of the main paper also will be further justified.

\subsection{Implementation Accelerations}
To speed up the evolutionary algorithms, we conduct pre-check procedures on the mutated set $S'$ before evaluating its oracle value $f(S')$. The pre-check procedures include two tasks:
\begin{itemize}
    \item \textit{Task 1}: Check if $S'$ stayed the same as $S$ (the set input to Algorithm \ref{alg:mutate}). 
    \item \textit{Task 2}: Check if $S'$ has already been evaluated in a previous iteration. This task can be done by maintaining a list/hashmap of all intermediate solution sets that have been evaluated, but it is memory and time-inefficient. So, we consider an alternative data structure -- bloom filter \cite{bloom1970space}.
\end{itemize}
If the pre-check procedures return ``yes" in any of the two tasks, $S'$ should not be evaluated.

Next, we argue why the probability of the first task returns ``yes" with non-trivial probability. We also give the implementation details of a bloom filter for Task 2.

\subsubsection{Probability that a set stays the same after mutation.} In the mutation procedure (Algorithm \ref{alg:mutate}), every element of a set $S$ is flipped membership with probability $1/n$ independently. Therefore, the probability that no element was flipped membership is 
 $\rho = (1-\frac{1}{n})^n$. If $n>1$, then
 $\rho \in [\frac{1}{4}, \frac{1}{e}]$, which is non-trivial.  

\subsubsection{Bloom Filter Implementation Details}
A Bloom Filter is a space-efficient probabilistic data structure used for membership testing \cite{bloom1970space}. When implemented with multiple hash functions, it becomes more robust. Bloom Filters are used in scenarios where memory efficiency is crucial, such as network routers, spell checkers, and databases to quickly filter out unnecessary queries.

In evolutionary algorithms discussed in this paper, we want to know if a mutated set $S'$ has been evaluated in the previous iterations. However, maintaining all evaluated sets in the main memory costs a lot of memory and the query time is pricy. A good approach is to maintain a bloom filter that can quickly give us the answer. The bloom filter should have at least three methods: Initialization(), Check($S'$), and Insert($S'$). We present the details of how we implement bloom filters in the evolutionary algorithms in this paper.

\begin{enumerate}
    \item \textbf{Initialization()}\\
    This function initializes the bloom filter data structures and hash functions.
    Allocate a bit array $B$ of size $m = 16T$ where $T$ is the number of iterations input to the algorithm. Initialize $B[j] = 0$ for all $j \in \{0,\cdots,m-1\}$. In C++ and Java, each unit of a bit array takes 1 bit of memory. Choose $k = 16 \ln 2$ independent hash functions, each producing values in the $[0, m-1]$ range. We use random hash functions defined by $h_i(S) = (a_i \cdot S + b_i) \% m$ where $i \in \{1,2,\cdots, k\}$, $a_i$ and $b_i$ are random positive integers.

    The initialization is called only once at the beginning of the algorithm.
    
    \item \textbf{Check($S'$)}\\
    Before evaluating a mutated set $S'$, we check if $S'$ is in the Bloom Filter. If so, we don't need to evaluate $f(S')$, since oracle computation is time-consuming. 
    
    The Check procedure behaves: Apply each of the $k$ hash functions to $S'$ ($S'$ is a set or 0-1 string but can be converted to an integer). Check if $B[h_i(S')] = 1$ for all $i \in \{1,2,\cdots, k\}$. If true, then $S'$ has probably been evaluated during some previous iteration, and the Check returns ``yes". If there exists $i \in \{1,2,\cdots, k\}$ such that $h_i(S') = 0$, then $S'$ has definitely not been evaluated and the Check returns ``no". 
    \item \textbf{Insert($S'$)}\\
    If the Check procedure returns ``no" given a set $S'$, we must evaluate $f(S')$. After the evaluation, we will insert $S'$ to the bloom filter -- meaning $S'$ has been evaluated.
   
    To insert $S'$ into the bloom filter, we apply each of the $k$ hash functions to $S'$. Set $B[h_i(S')] = 1$ for all $i \in \{1,2,\cdots, k\}$.
\end{enumerate}

\paragraph{False Positive Rate.} The probability of a false positive is determined by the number of hash functions ($k$), the size of the bit array ($m$), and the number of elements ($n$). We denote the ratio of   the bit array size and set size as $\tau=\frac{m}{n}$, also known as bits-per-element. A false positive is caused by the collision of hash functions when distinct sets have the same hash values. In theory, the optimal number of hash functions to achieve the smallest false positive rate is $k= \tau \cdot \ln 2$ and the corresponding false positive rate is $0.618^\tau$. 
Notably, if $\tau = 16$, the false positive rate is approximately $0.0004$.

\subsubsection{Effectiveness Study}
We study the effectiveness of the accelerations we discussed in this paper. We run {\sc evo-SMC} on the Film-Trust network with the application of influence maximization with costs. In the experiments, we run  various budgets ($\beta=20,35$ and $50$). During a run of an experiment, we equally set 15 breakpoints to record the results. (If the an algorithm needs $T$ iterations, we output the intermediate results at iterations $\{\frac{1}{15}T,\frac{2}{15}T, \cdots, T\}$). At each breakpoint, we record ``the ratio of mutations that stay the same" (Task 1) and the ``ratio of mutations that have been evaluated" (Task 2).  

The results show that, for every choice of $\beta$, ``the ratio of mutation staying the same" falls between 0.36 and 0.37, which is consistent with the theoretical results. Moreover, the bloom filters can effectively prevent repetitive evaluations (17\% to 25\% of the mutations). 

\subsubsection{Oracle Evaluations for Augmentations (line 19 of Algorithm \ref{alg:bpo})}
A solution can be augmented at line 19 of Algorithm \ref{alg:bpo}, which can take at most $n-1$ oracle evaluations to function $f$. In the last paragraph of Section \ref{sec:experiements} in the main paper, we observed that, in the experiments, the budget uses up quickly, and only a few portions of elements need to be evaluated. Therefore, we also present where the observation is from by studying the experiment on the Film-Trust network (same experimental settings as in the ``Effectiveness Study").

At each breakpoint, we recorded the cost of the set that was selected for mutation. The plot is shown in Figure \ref{fig:costs-app}. We observe that the evolutionary algorithms use up the budget fast, at approximately the halfway mark. There is no significant difference among various budget settings. In addition, we also focus on the percentage of elements we still need to consider for a solution augmentation (at line 19 of Algorithm \ref{alg:bpo}).  Figure \ref{fig:remaining-app} plots the ratio of elements required to evaluate. With fewer budgets left, the number of feasible elements for augmentations decreases rapidly. Smaller cost constraints result in a more rapid decrease. 

\begin{figure}[H]
\centering
\subfigure[]{
\centering
\includegraphics[width=0.4\textwidth]{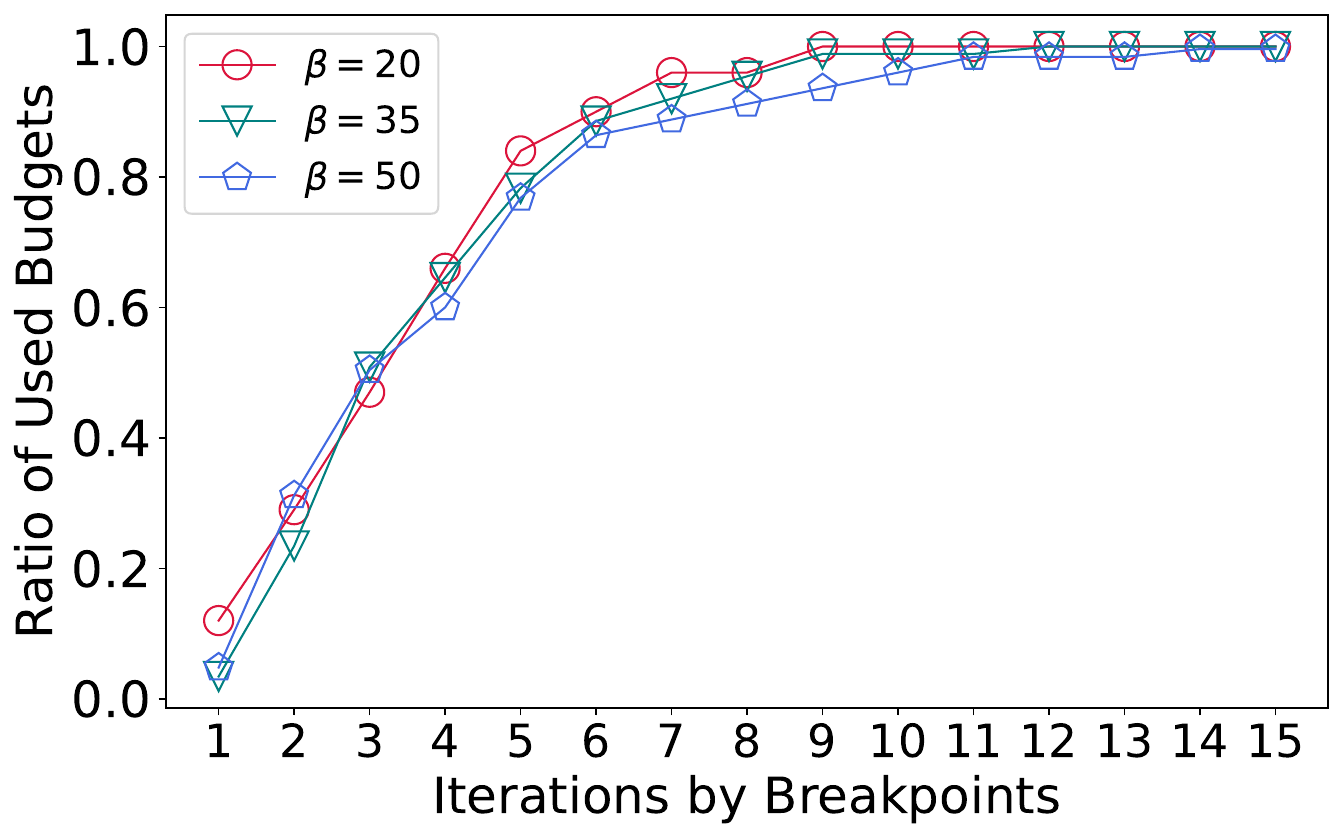}
% \caption{}
\label{fig:costs-app}
}
~~~~
\subfigure[]{
\centering
\includegraphics[width=0.4\textwidth]{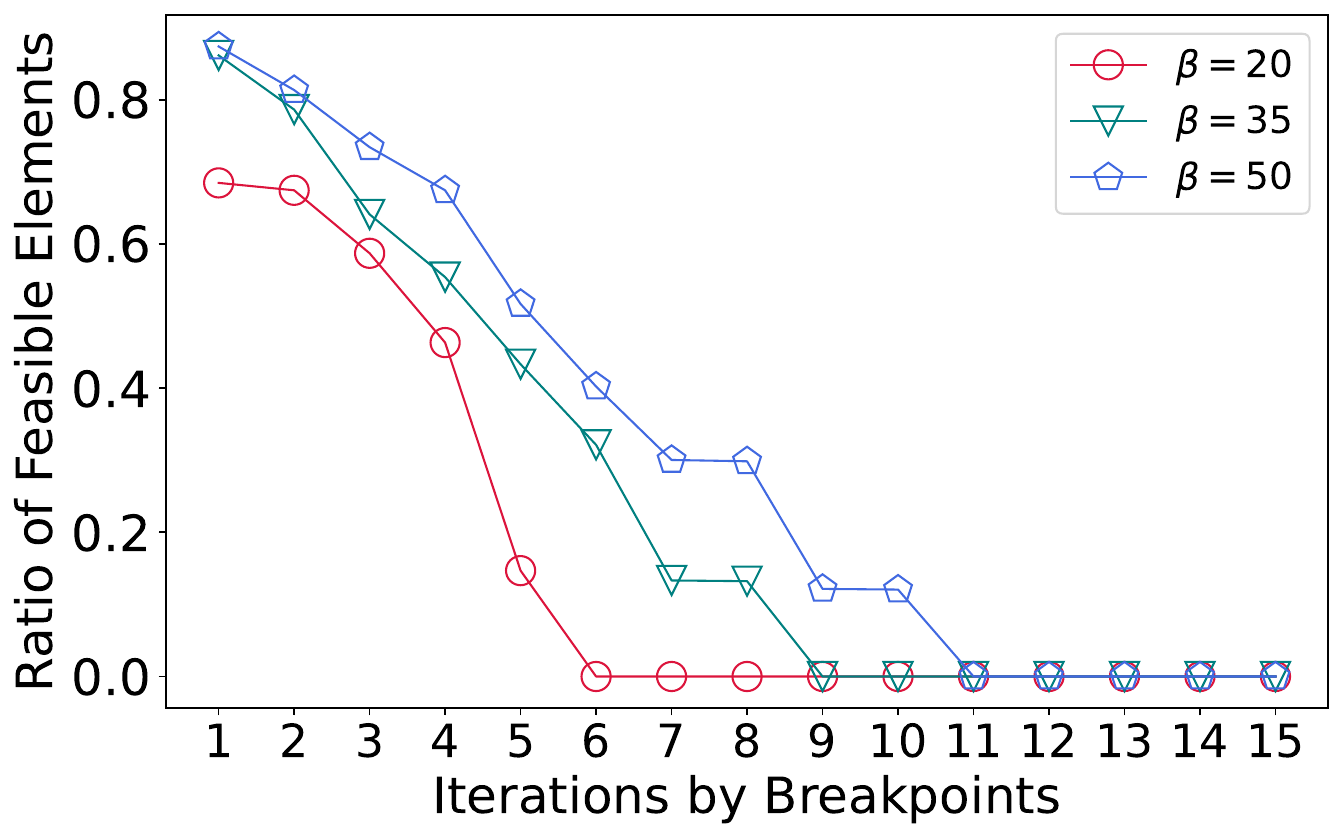}
% \caption{}
\label{fig:remaining-app}
}
\caption{At each breakpoint, we record the cost of the selected set $S$ and compute $c(S)/\beta$ in Fig. \ref{fig:costs-app}. We also retrieve $G_i$ and compute the number of feasible elements that can be considered to augment it.}
\end{figure}

\subsection{Additional Experimental Results and Analysis}
We present additional experiments with various stochasticity probabilities $p$ on the applications of influence maximization and vertex cover with costs. 

The results on Influence Maximization with various of $\epsilon$ are illustrated in Figure \ref{fig:film-eps-app} and \ref{fig:fb-eps-app}. Additional vertex cover results are shown in Figure \ref{fig:protein-q5-eps-app}, \ref{fig:protein-q10-eps-app} and Figure \ref{fig:email-q5-eps-app}, \ref{fig:email-q10-eps-app}. Our algorithms consistently outperform EAMC and Greedy+Max after a few additional evaluations. For $\epsilon=0.5$, in most scenarios, the algorithm {\sc st-evo-SMC} grows fast and results in a good quality solution. Nonetheless, {\sc st-evo-SMC} has better performance with $\epsilon=0.1$; it performs the best in general. Since the theoretical runtime bound of {\sc st-evo-SMC}, $T=O(nK_\beta \log (1/\epsilon)/p)$ is proportional to $\log (1/\epsilon)$, one would be suggested setting a relatively small $\epsilon$ to guarantee a good performance. At the same time, it will not sacrifice the running time a lot.

The results on Influence Maximization with various stochasticity probability $p$ are illustrated in Figure \ref{fig:film-p-app} and \ref{fig:fb-p-app}. Additional vertex cover results are shown in Figure \ref{fig:protein-q5-p-app}, \ref{fig:protein-q10-p-app} and Figure \ref{fig:email-q5-p-app}, \ref{fig:email-q10-p-app}. We can see that the with larger $p$, the increase rates are faster (Figure \ref{fig:film-p-app}, \ref{fig:email-q10-eps-app}, \ref{fig:email-q5-eps-app}, etc. ). With both $p=0.2$ and $p=0.5$, {\sc st-evo-SMC} is more efficient and can achieve good-quality results.

\begin{figure}[H]
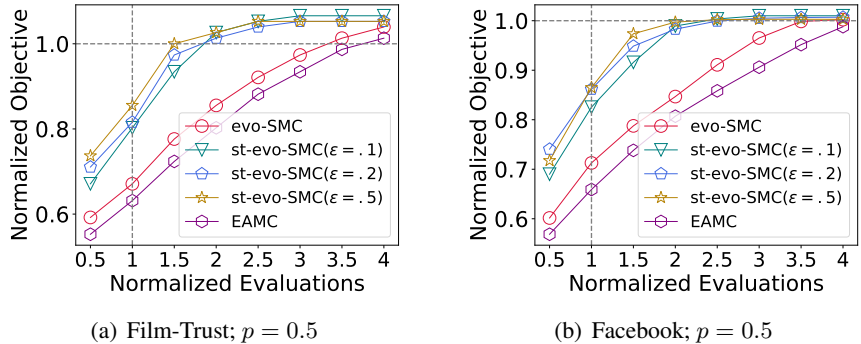

\centering
\subfigure[Film-Trust; $p=0.5$]{
\centering
\includegraphics[width=0.35\textwidth]{figures/film-trust-e.pdf}
% \caption{}
\label{fig:film-eps-app}
}
~~~~
\subfigure[Facebook; $p=0.5$]{
\centering
\includegraphics[width=0.35\textwidth]{figures/facebook-e.pdf}
% \caption{}
\label{fig:fb-eps-app}
}
\caption{Comparisons on Film-Trust and Facebook networks with various $\epsilon$.}
\end{figure}

\begin{figure}[t]
\centering
\subfigure[Film-Trust, $\epsilon=0.1$]{
\centering
\includegraphics[width=0.35\textwidth]{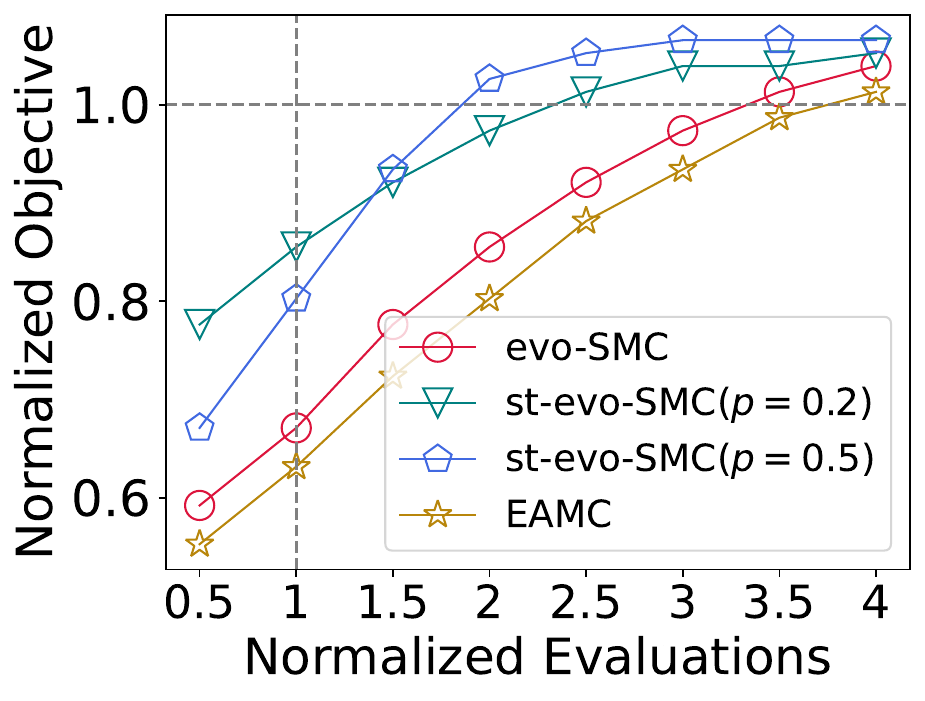}
% \caption{}
\label{fig:film-p-app}
}
~~~~
\subfigure[Facebook, $\epsilon=0.1$]{
\centering
\includegraphics[width=0.35\textwidth]{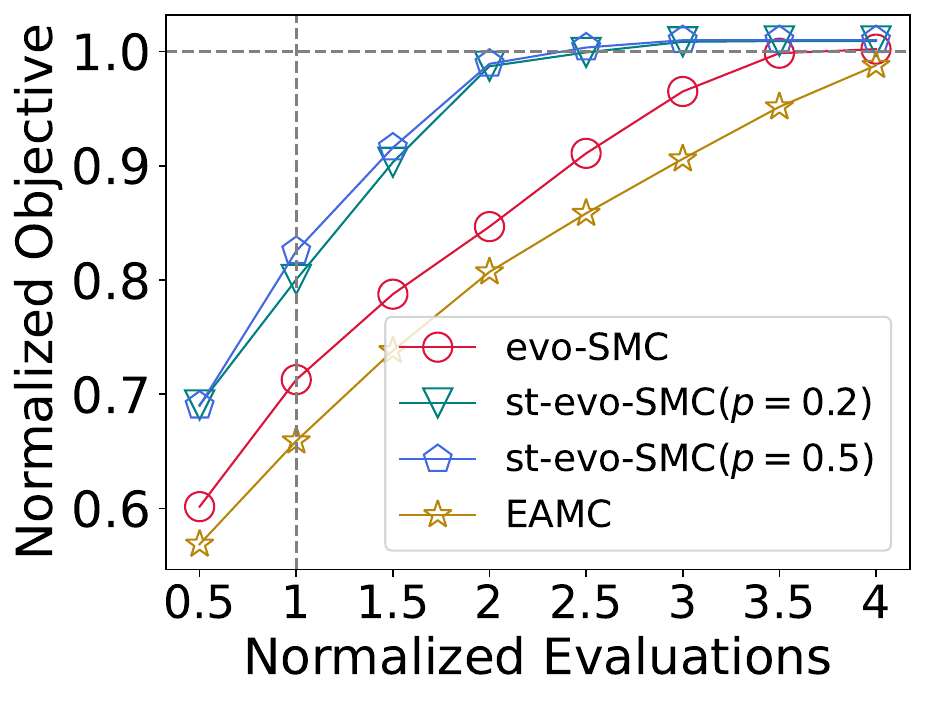}
% \caption{}
\label{fig:fb-p-app}
}
\caption{Comparisons on Film-Trust and Facebook networks with various $p$}
\end{figure}

\begin{figure}[H]
\centering
\subfigure[Protein; $q=5$, $p=0.5$]{
\centering
\includegraphics[width=0.35\textwidth]{figures/protein-e-q5.pdf}
% \caption{}
\label{fig:protein-q5-eps-app}
}
~~~~
\subfigure[Protein; $q=10$, $p=0.5$]{
\centering
\includegraphics[width=0.35\textwidth]{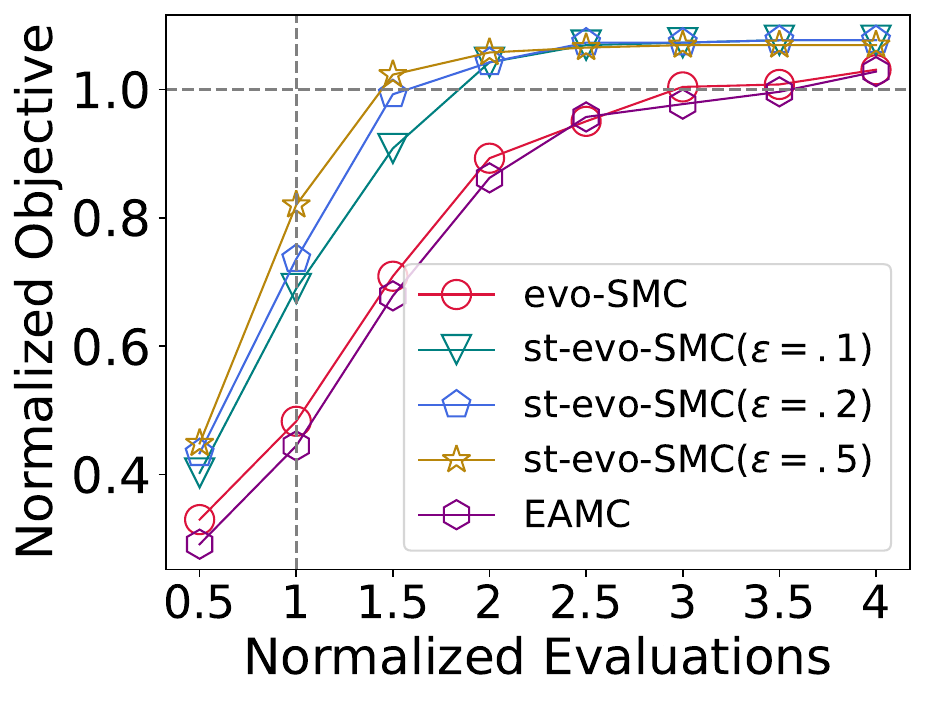}
% \caption{}
\label{fig:protein-q10-eps-app}
}
\caption{Comparisons on Protein network with various $\epsilon$ and fixed cost penalty $q=5$ and $q=10$.}
\end{figure}

\begin{figure}[t]
\centering
\subfigure[Protein; $q=5$, $\epsilon=0.1$]{
\centering
\includegraphics[width=0.35\textwidth]{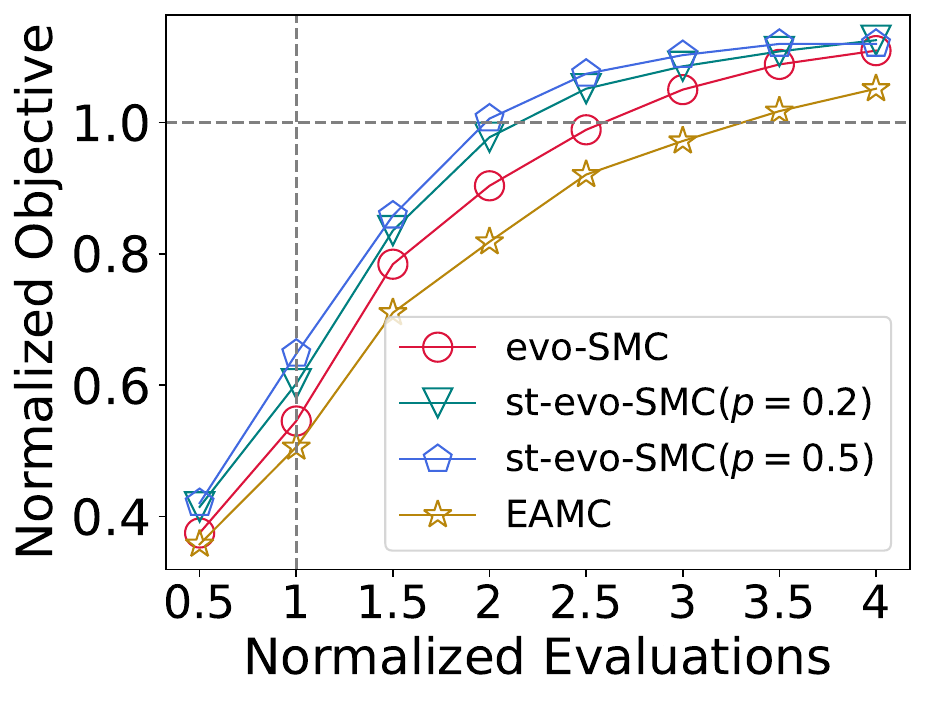}
% \caption{}
\label{fig:protein-q5-p-app}
}
~~~~
\subfigure[Protein; $q=10$, $\epsilon=0.1$]{
\centering
\includegraphics[width=0.35\textwidth]{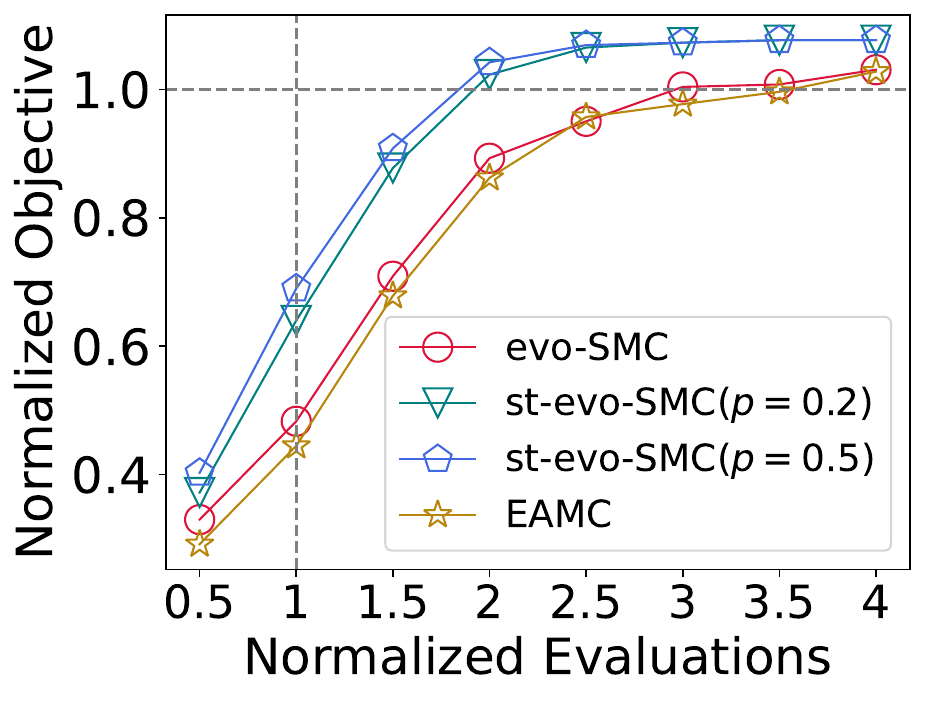}
% \caption{}
\label{fig:protein-q10-p-app}
}
\label{fig-protein-p}
\caption{Comparisons on Protein network with various $p$ and fixed cost penalty $q=5$ and $q=10$.}
\end{figure}

\begin{figure}[H]
\centering
\subfigure[Eu-Email; $q=5$, $p=0.5$]{
\centering
\includegraphics[width=0.35\textwidth]{figures/email-e-q5.pdf}
% \caption{}
\label{fig:email-q5-eps-app}
}
~~~~
\subfigure[Eu-Email; $q=10$, $p=0.5$]{
\centering
\includegraphics[width=0.35\textwidth]{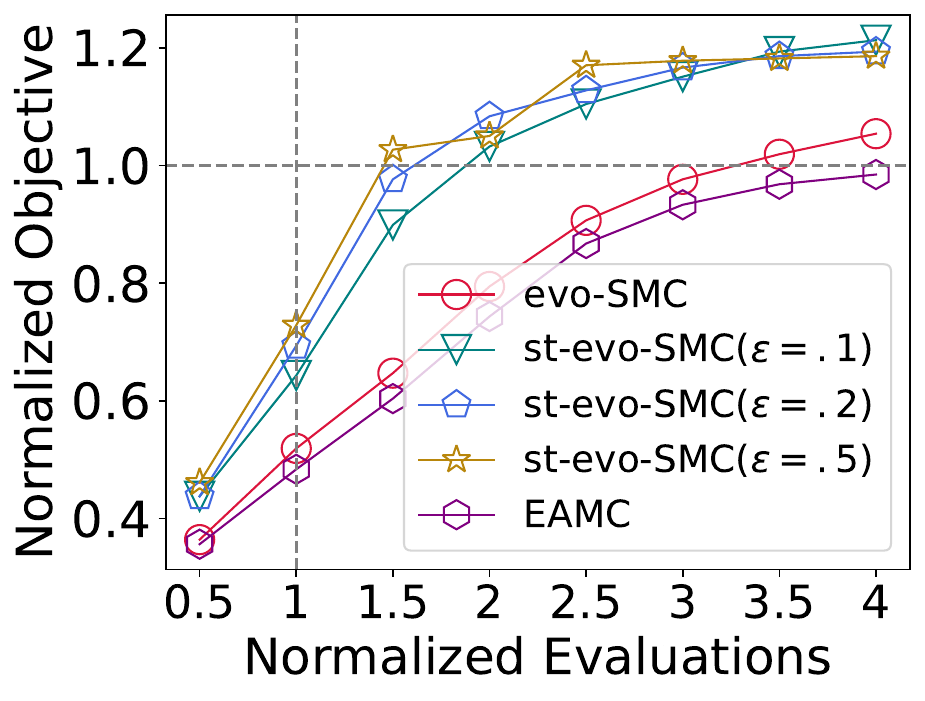}
% \caption{}
\label{fig:email-q10-eps-app}
}
\caption{Comparisons on Eu-Email network with various $\epsilon$ and fixed cost penalty $q=5$ and $q=10$.}
\end{figure}

\begin{figure}[H]
\centering
\subfigure[Eu-Email; $q=5$, $\epsilon=0.1$]{
\centering
\includegraphics[width=0.35\textwidth]{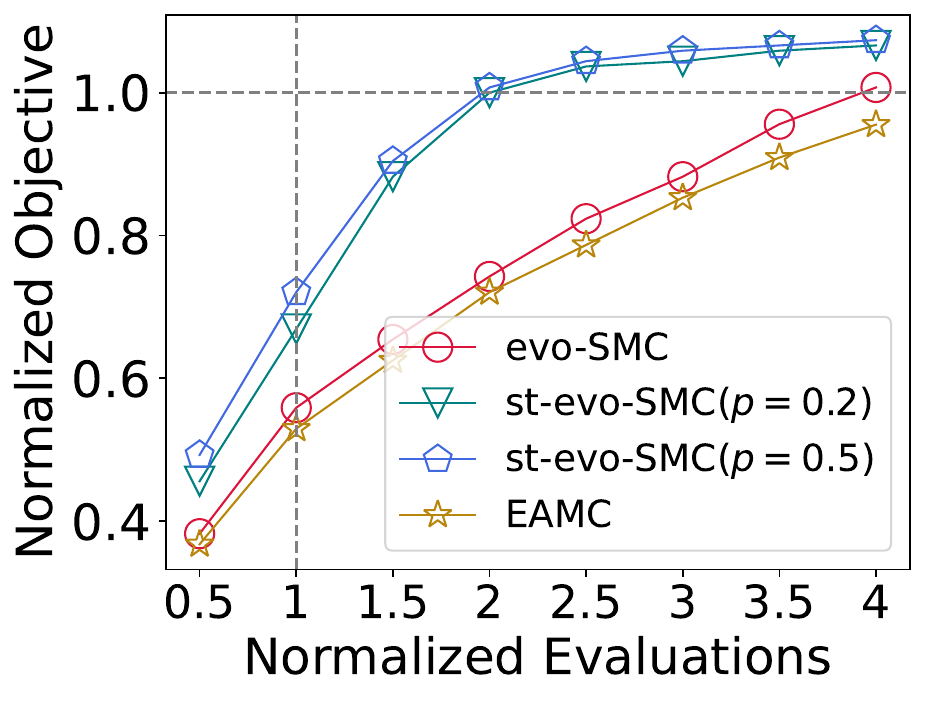}
% \caption{}
\label{fig:email-q5-p-app}
}
~~~~
\subfigure[Eu-Email; $q=10$, $\epsilon=0.1$]{
\centering
\includegraphics[width=0.35\textwidth]{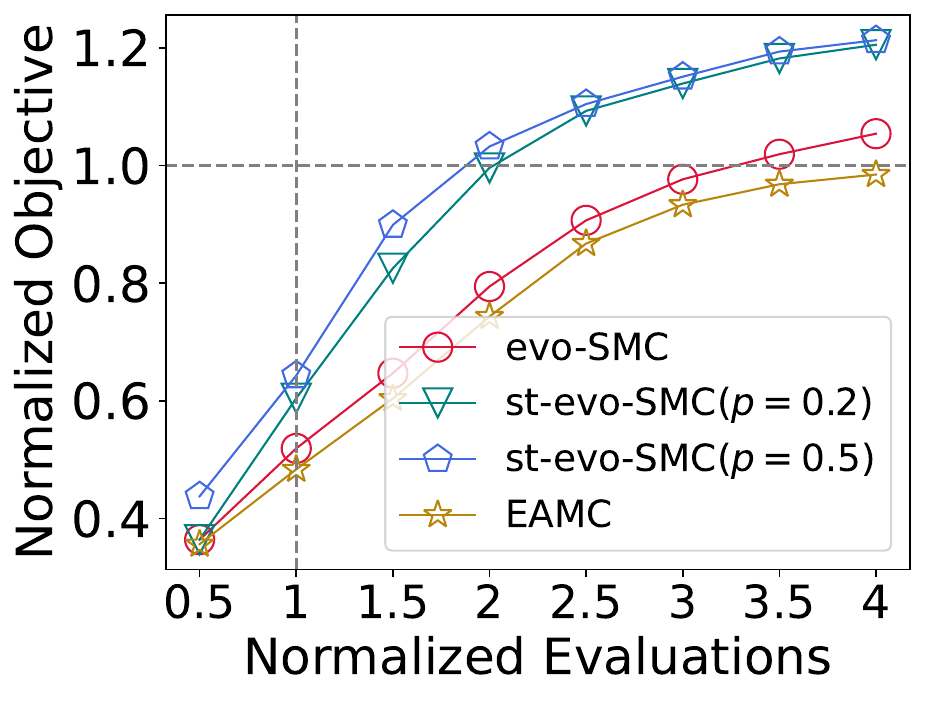}
% \caption{}
\label{fig:email-q10-p-app}
}
\caption{Comparisons on Eu-Email network with various $p$ and fixed cost penalty $q=5$ and $q=10$.}
\end{figure}

\bibliography{references}

\begin{thebibliography}{36}
\providecommand{\natexlab}[1]{#1}
\providecommand{\url}[1]{\texttt{#1}}
\expandafter\ifx\csname urlstyle\endcsname\relax
  \providecommand{\doi}[1]{doi: #1}\else
  \providecommand{\doi}{doi: \begingroup \urlstyle{rm}\Url}\fi

\bibitem[Badanidiyuru \& Vondr{\'a}k(2014)Badanidiyuru and Vondr{\'a}k]{badanidiyuru2014fast}
Ashwinkumar Badanidiyuru and Jan Vondr{\'a}k.
\newblock Fast algorithms for maximizing submodular functions.
\newblock In \emph{Proceedings of the twenty-fifth annual ACM-SIAM symposium on Discrete algorithms}, pp.\  1497--1514. SIAM, 2014.

\bibitem[Bateni et~al.(2019)Bateni, Chen, Esfandiari, Fu, Mirrokni, and Rostamizadeh]{bateni2019categorical}
MohammadHossein Bateni, Lin Chen, Hossein Esfandiari, Thomas Fu, Vahab Mirrokni, and Afshin Rostamizadeh.
\newblock Categorical feature compression via submodular optimization.
\newblock In \emph{International Conference on Machine Learning}, pp.\  515--523. PMLR, 2019.

\bibitem[Bian et~al.(2020)Bian, Feng, Qian, and Yu]{bian2020efficient}
Chao Bian, Chao Feng, Chao Qian, and Yang Yu.
\newblock An efficient evolutionary algorithm for subset selection with general cost constraints.
\newblock In \emph{Proceedings of the AAAI Conference on Artificial Intelligence}, volume~34, pp.\  3267--3274, 2020.

\bibitem[Bian et~al.(2021)Bian, Qian, Neumann, and Yu]{bian2021fast}
Chao Bian, Chao Qian, Frank Neumann, and Yang Yu.
\newblock Fast pareto optimization for subset selection with dynamic cost constraints.
\newblock In \emph{IJCAI}, pp.\  2191--2197, 2021.

\bibitem[Bloom(1970)]{bloom1970space}
Burton~H Bloom.
\newblock Space/time trade-offs in hash coding with allowable errors.
\newblock \emph{Communications of the ACM}, 13\penalty0 (7):\penalty0 422--426, 1970.

\bibitem[Chen et~al.(2022)Chen, Feng, Cao, Zeng, and Hou]{chen2022budgeted}
Xuefeng Chen, Liang Feng, Xin Cao, Yifeng Zeng, and Yaqing Hou.
\newblock Budgeted sequence submodular maximization.
\newblock In Luc~De Raedt (ed.), \emph{Proceedings of the Thirty-First International Joint Conference on Artificial Intelligence, {IJCAI} 2022, Vienna, Austria, 23-29 July 2022}, pp.\  4733--4739. ijcai.org, 2022.
\newblock \doi{10.24963/IJCAI.2022/656}.
\newblock URL \url{https://doi.org/10.24963/ijcai.2022/656}.

\bibitem[Crawford(2019)]{crawford2019faster}
Victoria~G. Crawford.
\newblock An efficient evolutionary algorithm for minimum cost submodular cover.
\newblock In Sarit Kraus (ed.), \emph{Proceedings of the Twenty-Eighth International Joint Conference on Artificial Intelligence, {IJCAI} 2019, Macao, China, August 10-16, 2019}, pp.\  1227--1233. ijcai.org, 2019.

\bibitem[Do \& Neumann(2021)Do and Neumann]{do2021pareto}
Anh~Viet Do and Frank Neumann.
\newblock Pareto optimization for subset selection with dynamic partition matroid constraints.
\newblock In \emph{Proceedings of the AAAI Conference on Artificial Intelligence}, volume~35, pp.\  12284--12292, 2021.

\bibitem[El~Halabi et~al.(2022)El~Halabi, Srinivas, and Lacoste-Julien]{el2022data}
Marwa El~Halabi, Suraj Srinivas, and Simon Lacoste-Julien.
\newblock Data-efficient structured pruning via submodular optimization.
\newblock \emph{Advances in Neural Information Processing Systems}, 35:\penalty0 36613--36626, 2022.

\bibitem[Ene \& Nguyen(2019)Ene and Nguyen]{EneN19}
Alina Ene and Huy~L. Nguyen.
\newblock A nearly-linear time algorithm for submodular maximization with a knapsack constraint.
\newblock In Christel Baier, Ioannis Chatzigiannakis, Paola Flocchini, and Stefano Leonardi (eds.), \emph{46th International Colloquium on Automata, Languages, and Programming, {ICALP} 2019, July 9-12, 2019, Patras, Greece}, volume 132 of \emph{LIPIcs}, pp.\  53:1--53:12. Schloss Dagstuhl - Leibniz-Zentrum f{\"{u}}r Informatik, 2019.

\bibitem[Feldman et~al.(2022)Feldman, Nutov, and Shoham]{feldman2022practical}
Moran Feldman, Zeev Nutov, and Elad Shoham.
\newblock Practical budgeted submodular maximization.
\newblock \emph{Algorithmica}, pp.\  1--40, 2022.

\bibitem[Friedrich \& Neumann(2015)Friedrich and Neumann]{friedrich2015maximizing}
Tobias Friedrich and Frank Neumann.
\newblock Maximizing submodular functions under matroid constraints by evolutionary algorithms.
\newblock \emph{Evolutionary computation}, 23\penalty0 (4):\penalty0 543--558, 2015.

\bibitem[Harshaw et~al.(2019)Harshaw, Feldman, Ward, and Karbasi]{harshaw2019submodular}
Chris Harshaw, Moran Feldman, Justin Ward, and Amin Karbasi.
\newblock Submodular maximization beyond non-negativity: Guarantees, fast algorithms, and applications.
\newblock In \emph{International Conference on Machine Learning}, pp.\  2634--2643. PMLR, 2019.

\bibitem[Iyer \& Bilmes(2013)Iyer and Bilmes]{iyer2013submodular}
Rishabh~K Iyer and Jeff~A Bilmes.
\newblock Submodular optimization with submodular cover and submodular knapsack constraints.
\newblock \emph{Advances in neural information processing systems}, 26, 2013.

\bibitem[Jin et~al.(2021)Jin, Yang, Yang, Shi, Huang, and Xiao]{jin2021unconstrained}
Tianyuan Jin, Yu~Yang, Renchi Yang, Jieming Shi, Keke Huang, and Xiaokui Xiao.
\newblock Unconstrained submodular maximization with modular costs: Tight approximation and application to profit maximization.
\newblock \emph{Proceedings of the VLDB Endowment}, 14\penalty0 (10):\penalty0 1756--1768, 2021.

\bibitem[Kempe et~al.(2003)Kempe, Kleinberg, and Tardos]{kempe2003maximizing}
David Kempe, Jon Kleinberg, and {\'E}va Tardos.
\newblock Maximizing the spread of influence through a social network.
\newblock In \emph{Proceedings of the ninth ACM SIGKDD international conference on Knowledge discovery and data mining}, pp.\  137--146, 2003.

\bibitem[Khuller et~al.(1999)Khuller, Moss, and Naor]{khuller1999budgeted}
Samir Khuller, Anna Moss, and Joseph~Seffi Naor.
\newblock The budgeted maximum coverage problem.
\newblock \emph{Information processing letters}, 70\penalty0 (1):\penalty0 39--45, 1999.

\bibitem[Krause \& Guestrin(2005)Krause and Guestrin]{krause2005note}
Andreas Krause and Carlos Guestrin.
\newblock \emph{A note on the budgeted maximization of submodular functions}.
\newblock Citeseer, 2005.

\bibitem[Kunegis(2013)]{kunegis2013konect}
J{\'e}r{\^o}me Kunegis.
\newblock Konect: the koblenz network collection.
\newblock In \emph{Proceedings of the 22nd international conference on world wide web}, pp.\  1343--1350, 2013.

\bibitem[Leskovec \& Mcauley(2012)Leskovec and Mcauley]{leskovec2012learning}
Jure Leskovec and Julian Mcauley.
\newblock Learning to discover social circles in ego networks.
\newblock \emph{Advances in neural information processing systems}, 25, 2012.

\bibitem[Leskovec et~al.(2007)Leskovec, Kleinberg, and Faloutsos]{leskovec2007graph}
Jure Leskovec, Jon Kleinberg, and Christos Faloutsos.
\newblock Graph evolution: Densification and shrinking diameters.
\newblock \emph{ACM transactions on Knowledge Discovery from Data (TKDD)}, 1\penalty0 (1):\penalty0 2--es, 2007.

\bibitem[Li et~al.(2023)Li, Mehr, and Horowitz]{li2023submodularity}
Ruolin Li, Negar Mehr, and Roberto Horowitz.
\newblock Submodularity of optimal sensor placement for traffic networks.
\newblock \emph{Transportation Research Part B: Methodological}, 171:\penalty0 29--43, 2023.

\bibitem[Li et~al.(2022)Li, Feldman, Kazemi, and Karbasi]{li2022submodular}
Wenxin Li, Moran Feldman, Ehsan Kazemi, and Amin Karbasi.
\newblock Submodular maximization in clean linear time.
\newblock In \emph{Advances in Neural Information Processing Systems}, 2022.

\bibitem[Mirzasoleiman et~al.(2015)Mirzasoleiman, Badanidiyuru, Karbasi, Vondr{\'a}k, and Krause]{mirzasoleiman2015lazier}
Baharan Mirzasoleiman, Ashwinkumar Badanidiyuru, Amin Karbasi, Jan Vondr{\'a}k, and Andreas Krause.
\newblock Lazier than lazy greedy.
\newblock In \emph{Proceedings of the AAAI Conference on Artificial Intelligence}, volume~29, 2015.

\bibitem[Nemhauser et~al.(1978)Nemhauser, Wolsey, and Fisher]{nemhauser1978analysis}
George~L Nemhauser, Laurence~A Wolsey, and Marshall~L Fisher.
\newblock An analysis of approximations for maximizing submodular set functions—i.
\newblock \emph{Mathematical programming}, 14\penalty0 (1):\penalty0 265--294, 1978.

\bibitem[Padmanabhan et~al.(2023)Padmanabhan, Zhu, Basu, and Pavan]{padmanabhan2023maximizing}
Madhavan~R Padmanabhan, Yanhui Zhu, Samik Basu, and A~Pavan.
\newblock Maximizing submodular functions under submodular constraints.
\newblock In \emph{Uncertainty in Artificial Intelligence}, pp.\  1618--1627. PMLR, 2023.

\bibitem[Qian(2021)]{qian2021multiobjective}
Chao Qian.
\newblock Multiobjective evolutionary algorithms are still good: Maximizing monotone approximately submodular minus modular functions.
\newblock \emph{Evolutionary Computation}, 29\penalty0 (4):\penalty0 463--490, 2021.

\bibitem[Qian et~al.(2015)Qian, Yu, and Zhou]{qian2015subset}
Chao Qian, Yang Yu, and Zhi-Hua Zhou.
\newblock Subset selection by pareto optimization.
\newblock \emph{Advances in neural information processing systems}, 28, 2015.

\bibitem[Qian et~al.(2017)Qian, Shi, Yu, and Tang]{qian2017subset}
Chao Qian, Jing-Cheng Shi, Yang Yu, and Ke~Tang.
\newblock On subset selection with general cost constraints.
\newblock In \emph{IJCAI}, volume~17, pp.\  2613--2619, 2017.

\bibitem[Roostapour et~al.(2022)Roostapour, Neumann, Neumann, and Friedrich]{roostapour2022pareto}
Vahid Roostapour, Aneta Neumann, Frank Neumann, and Tobias Friedrich.
\newblock Pareto optimization for subset selection with dynamic cost constraints.
\newblock \emph{Artificial Intelligence}, 302:\penalty0 103597, 2022.

\bibitem[Stelzl et~al.(2005)Stelzl, Worm, Lalowski, Haenig, Brembeck, Goehler, Stroedicke, Zenkner, Schoenherr, Koeppen, et~al.]{stelzl2005human}
Ulrich Stelzl, Uwe Worm, Maciej Lalowski, Christian Haenig, Felix~H Brembeck, Heike Goehler, Martin Stroedicke, Martina Zenkner, Anke Schoenherr, Susanne Koeppen, et~al.
\newblock A human protein-protein interaction network: a resource for annotating the proteome.
\newblock \emph{Cell}, 122\penalty0 (6):\penalty0 957--968, 2005.

\bibitem[Sviridenko(2004)]{sviridenko2004note}
Maxim Sviridenko.
\newblock A note on maximizing a submodular set function subject to a knapsack constraint.
\newblock \emph{Operations Research Letters}, 32\penalty0 (1):\penalty0 41--43, 2004.

\bibitem[Tang et~al.(2020)Tang, Tang, Lim, Han, Li, and Yuan]{abs-2008-05391}
Jing Tang, Xueyan Tang, Andrew Lim, Kai Han, Chongshou Li, and Junsong Yuan.
\newblock Revisiting modified greedy algorithm for monotone submodular maximization with a knapsack constraint.
\newblock \emph{CoRR}, abs/2008.05391, 2020.

\bibitem[Yaroslavtsev et~al.(2020)Yaroslavtsev, Zhou, and Avdiukhin]{yaroslavtsev2020bring}
Grigory Yaroslavtsev, Samson Zhou, and Dmitrii Avdiukhin.
\newblock “bring your own greedy”+ max: near-optimal 1/2-approximations for submodular knapsack.
\newblock In \emph{International Conference on Artificial Intelligence and Statistics}, pp.\  3263--3274. PMLR, 2020.

\bibitem[Zheng et~al.(2013)Zheng, Liu, and Hsieh]{zheng2013u}
Yu~Zheng, Furui Liu, and Hsun-Ping Hsieh.
\newblock U-air: When urban air quality inference meets big data.
\newblock In \emph{Proceedings of the 19th ACM SIGKDD international conference on Knowledge discovery and data mining}, pp.\  1436--1444, 2013.

\bibitem[Zhu et~al.(2024)Zhu, Hu, Kuo, and Liu]{zhu2024score+}
Yanhui Zhu, Fang Hu, Lei~Hsin Kuo, and Jia Liu.
\newblock Scoreh+: A high-order node proximity spectral clustering on ratios-of-eigenvectors algorithm for community detection.
\newblock \emph{IEEE Transactions on Big Data}, 10\penalty0 (3):\penalty0 301--312, 2024.
\newblock \doi{10.1109/TBDATA.2023.3346715}.

\end{thebibliography}

\bibliographystyle{tmlr}

\end{document}